\documentclass[journal,10pt,twocolumn,twoside]{IEEEtran}

\usepackage{amsfonts}
\usepackage{amsmath}
\usepackage{amssymb}
\usepackage[usenames,dvipsnames]{color}
\usepackage{lscape}
\usepackage{epsf}
\usepackage{graphicx}
\usepackage{psfrag}

\newtheorem{theorem}{\indent Theorem}[section]
\newtheorem{lemma}[theorem]{\indent Lemma}
\newtheorem{corollary}[theorem]{\indent Corollary}

\newtheorem{EXAMPLE}{\indent Example}[section]
\newtheorem{definition}{\indent Definition}[section]

\newenvironment{example}{\begin{EXAMPLE}\rm}{\rm\end{EXAMPLE}}

\def\bu{\bar{u}}

\newcommand{\defeq}{:=}

\newcommand{\sff} {\mathsf{f}}

\newcommand{\cM}{{\mathcal{M}}}

\newcommand{\Code}{{\mathbb{C}}}

\newcommand{\coeff}{{\mbox{Coeff }}}

\newcommand{\bldalpha}{{\mbox{\boldmath $\alpha$}}}
\newcommand{\bldsmallalpha}{{\mbox{\scriptsize \boldmath $\alpha$}}}

\newcommand{\bldbeta}{{\mbox{\boldmath $\beta$}}}
\newcommand{\bldsmallbeta}{{\mbox{\scriptsize \boldmath $\beta$}}}

\newcommand{\bldc}{{\mbox{\boldmath $c$}}}
\newcommand{\bldcc}{{\mbox{\scriptsize \boldmath $c$}}}

\newcommand{\bldepsilon}{{\mbox{\boldmath $\epsilon$}}}
\newcommand{\bldsmallepsilon}{{\mbox{\scriptsize \boldmath $\epsilon$}}}

\newcommand{\asympequalm}{\doteq}
\newcommand{\asympequaln}{\doteq}
\newcommand{\asympequal}{\doteq}

    \def\squarebox#1{\hbox to #1{\hfill\vbox to #1{\vfill}}}

\newlength{\Algwidth}

\title{Spectral Shape of Doubly-Generalized LDPC Codes: Efficient and Exact Evaluation
\thanks{%
    This work was supported in part by the EC under Seventh FP grant agreement n. 288502 CONCERTO, and in part by Science Foundation Ireland (grants 07/SK/I1252b and 11/RFP.1/ECE/3206). The material in this paper was presented in part at the IEEE Global Communications Conference (GLOBECOM '09), Honolulu, Hawaii, Nov./Dec. 2009, in part at the IEEE International Conference on Communications (ICC '10), Cape Town, South Africa, May 2010, and in part at the 2nd International Symposium on Applied Sciences in Biomedical and Communication Technologies (ISABEL '09), Bratislava, Slovakia, Nov. 2009.  
    \newline 
		M. F. Flanagan is with the School of Electrical, Electronic and Communications Engineering, University College Dublin, Belfield, Dublin 4, Ireland (e-mail:mark.flanagan@ieee.org). 
		\newline
    E. Paolini and M. Chiani are with the Department of Electronic, Electrical and Information Engineering, University of Bologna, via Venezia 52, 47521 Cesena
(FC), Italy (e-mail:e.paolini@unibo.it, marco.chiani@unibo.it).
    \newline
    M. P. C. Fossorier is with ETIS ENSEA, UCP, CNRS UMR-8051, 6 avenue du Ponceau, 95014 Cergy Pontoise, France (e-mail: mfossorier@ieee.org). 
    }        
}

\author{Mark F. Flanagan,~\IEEEmembership{Senior Member,~IEEE,} Enrico Paolini,~\IEEEmembership{Member,~IEEE,} Marco~Chiani,~\IEEEmembership{Fellow,~IEEE,} and~Marc~P.~C.~Fossorier,~\IEEEmembership{Fellow,~IEEE}}

\begin{document}
\maketitle
\thispagestyle{empty}

\begin{abstract}
This paper analyzes the asymptotic exponent of the weight spectrum for irregular doubly-generalized
LDPC (D-GLDPC) codes. In the process, an efficient numerical technique for its evaluation is
presented, involving the solution of a $4 \times 4$ system of polynomial
equations. The expression is consistent with previous results, including the case where the normalized weight
or stopping set size tends to zero. The spectral shape is shown to admit a particularly simple form
in the special case where all variable nodes are repetition codes of the same degree, a case which
includes Tanner codes; for this case it is also shown how certain symmetry properties of the local
weight distribution at the CNs induce a symmetry in the overall weight spectral shape function.
Finally, using these new results, weight and stopping set size spectral shapes are evaluated for
some example generalized and doubly-generalized LDPC code ensembles. 
\end{abstract}

\section{Introduction}
Recently, the design and analysis of coding schemes representing generalizations of Gallager's low-density parity-check (LDPC) codes \cite{gallager63:low-density} has gained increasing attention. This interest is motivated above all by the search for coding schemes which offer a better compromise between waterfall and error floor performance than is currently offered by state-of-the-art LDPC codes.

In the Tanner graph of an LDPC code, any degree-$q$ variable node (VN) may be interpreted as a length-$q$ repetition code, i.e., as a $(q,1)$ linear block code. Similarly, any degree-$s$ check node (CN) may be interpreted as a length-$s$ single parity-check (SPC) code, i.e., as a $(s,s-1)$ linear block code. The first proposal of a class of linear block codes generalizing LDPC codes may be found in \cite{tanner81:recursive}, where it was suggested to replace each CN of a regular LDPC code with a generic linear block code, to enhance the overall minimum distance. The corresponding coding scheme is known as a regular generalized LDPC (GLDPC) code, or Tanner code, and a CN that is not a SPC code as a generalized CN. More recently, irregular GLDPC codes were considered (see for instance \cite{liva08:quasi_cyclic}). For such codes, the VNs exhibit different degrees and the CN set is composed of a mixture of different linear block codes.

A further generalization step is represented by doubly-generalized LDPC (D-GLDPC) codes \cite{wang06:D-GLDPC}. In a D-GLDPC code, not only the CNs but also the VNs may be represented by generic linear block codes. The VNs which are not repetition codes are called generalized VNs. The main motivation for introducing generalized VNs is to overcome some problems connected with the use of generalized CNs, such as an overall code rate loss which makes GLDPC codes interesting mainly for low code rate applications, and a loss in terms of decoding threshold (for a discussion on drawbacks of generalized CNs and on beneficial effects of generalized VNs we refer to \cite{miladinovic08:generalized} and \cite{paolini09:stability}, respectively).

A useful tool for analysis and design of LDPC codes and their generalizations is represented by the asymptotic exponent of the weight distribution. As usual in the literature, this exponent will be referred to as the \emph{growth rate of the weight distribution} or the \emph{weight spectral shape} of the ensemble, the two expressions being used interchangeably throughout this paper. The growth rate of the weight distribution was introduced in \cite{gallager63:low-density} to show that the minimum distance of a randomly generated regular LDPC code with a VN degree of at least three is a linear function of the codeword length with high probability. The same approach was taken in \cite{lentmaier99:generalized} and \cite{boutros99:generalized} to obtain related results on the minimum distance of subclasses of Tanner codes.

The growth rate of the weight distribution has been subsequently investigated for unstructured
ensembles of irregular LDPC codes. Works in this area are
\cite{litsyn02:on_ensembles}--\cite{di06:weight}.
In particular,
in \cite{di06:weight} a technique for evaluation of the growth rate of any (eventually
expurgated) irregular LDPC ensemble has been developed, based on Hayman's formula. The nonbinary weight distribution of nonbinary LDPC codes was analyzed in \cite{bennatan04:weight} and \cite{yang11:weight}, while the \emph{binary} weight distribution of nonbinary LDPC codes was derived in \cite{Kasai2008:nonbinary} and \cite{andriyanova09:weight}. Asymptotic weight enumerators of ensembles of irregular LDPC codes based on protographs and on multiple edge types have been derived in \cite{divsalar06:weight_enumerators}
and \cite{wang2009:multiedge,Kasai2009:multiedge}, respectively. The approach proposed in
\cite{divsalar06:weight_enumerators} has then been extended to protograph GLDPC codes and to
protograph D-GLDPC codes in \cite{abu-surra11:IEEE-IT} and \cite{wang08:ensemble_DGLDPC},
respectively. In contrast to the present work, the evaluation of the weight enumerators in
\cite{divsalar06:weight_enumerators}--\cite{wang08:ensemble_DGLDPC} require numerical solution of a
high-dimensional optimization problem.

In this paper, an analytical expression for the growth rate of the weight distribution of a general unstructured irregular ensemble of D-GLDPC codes is developed. The present work also extends to the fully-irregular case an expression for the growth rate obtained in \cite{paolini09:class} assuming a CN set composed of linear block codes all of the same type. In the process of this development, we obtain an efficient evaluation tool for computing the growth rate exactly. This tool always requires the solution of a $4 \times 4$ polynomial system of equations, \emph{regardless} of the number of VN types and CN types in the D-GLDPC ensemble. As shown through numerical examples, the proposed tool allows to obtain a precise plot of the growth rate with a low computational effort. The derived result may be regarded as a generalization to the D-GLDPC case of the corresponding result for LDPC codes which was derived in \cite{di06:weight}.

In \cite{flanagan09:IEEE-IT}, the growth rate of the weight distribution, $G(\alpha)$, was analyzed in the region of \emph{small} (fractional) codeword weight $\alpha$. A precise characterization of the class of code ensembles with \emph{good growth rate behavior} (i.e., having a negative initial slope of the growth rate curve -- typical codes from ensembles without this property have high error floors, even under \emph{maximum a posteriori} (MAP) decoding) was deduced in \cite{flanagan09:IEEE-IT} from the analysis therein.
In contrast to the formula developed in \cite{flanagan09:IEEE-IT}, the present paper proves an exact analytical expression for $G(\alpha)$ for \emph{any} $\alpha$. Note that a good weight spectral shape will not only have a negative initial slope, but the value $\alpha^*$ where it crosses the horizontal axis, called the \emph{critical exponent codeword weight ratio}, will also be large. For any code chosen from the ensemble, we expect ``exponentially few'' codewords of weight smaller than $\alpha^* n$, and ``exponentially many'' codewords of weight larger than $\alpha^* n$. Crucially, the results in the present paper allow us to evaluate $\alpha^*$ numerically for any given ensemble; this parameter then provides a first-order characterization of the weight distribution behavior (a reasonably large $\alpha^*$ is generally desirable).

As explained in Section~\ref{section:irregular_D_GLDPC}, assuming transmission over the binary
erasure channel (BEC) and iterative decoding, the developed formula is also valid for the asymptotic
exponent of the \emph{stopping set size} distribution upon replacing the local weight enumerating
function (WEF) of each VN and CN type with an appropriate polynomial function. 
This also allows for the evaluation of the \emph{critical exponent stopping set size ratio}, which is a good indicator of the error floor performance when using belief propagation decoding over the BEC.

A compact formula for the spectral shape is derived for the special case of a GLDPC code ensemble with a regular VN set and a hybrid CN set. Symmetry properties of the weight spectral shape are also investigated for this case; in particular, it is proved that the weight spectral shape function of such a GLDPC ensemble is symmetric w.r.t. normalized weight $\alpha = 1/2$ if the local WEF of each CN is a symmetric polynomial. This result establishes a connection between symmetry properties at a ``microscopic'' level (i.e., at the nodes of the Tanner graph) and symmetry of the ``macroscopic'' growth rate function.

The paper is organized as follows. Section~\ref{section:irregular_D_GLDPC} defines the D-GLDPC ensemble of interest, and introduces some definitions and notation pertaining to this ensemble. Section~\ref{section:growth_rate} presents the main result of the paper regarding the evaluation of the weight and stopping set size spectral shapes. Section~\ref{section:spectral_shape_formula} derives the spectral shapes of a family of check-hybrid GLDPC code ensembles as a corollary to the main result, and also identifies a sufficient condition for symmetry of the weight spectral shape for such ensembles. Section~\ref{section:proof_of_main_result} provides a proof of the main result of the paper, and Section~\ref{section:supporting_proofs} provides additional proofs for other results in the paper. Section~\ref{section:examples} provides some examples of spectral shapes of GLDPC and D-GLDPC codes, and Section~\ref{section:conclusion} concludes the paper.


\section{Preliminaries and Notation}\label{section:irregular_D_GLDPC}
A D-GLDPC code consists of a set of CNs and a set of VNs. Each of these nodes is associated with a
linear `local' code, and with each VN is also associated an \emph{encoder} (i.e., generator matrix).
Each node is equipped with a set of \emph{sockets} corresponding to the bits of the local codeword.
The VN sockets and the CNs sockets are connected together by edges in a one-to-one fashion; the
resulting graph is called the \emph{Tanner graph} of the code. A \emph{codeword} in such a D-GLDPC
code is defined as an assignment of values to the local information bits of each VN such that the
corresponding bit values induced on the Tanner graph edges (through local encoding at the VNs) cause
each CN to recognize a valid local codeword. An illustration of a simple D-GLDPC code is given in
Figure~\ref{cap:DGLDPC_example}, together with its codeword $[1 \; 0 \; 1 \; 0 \; 1 \; 0 \; 0 \; 0
\; 1]$. We next define a sequence $\{ \cM_n \}$ of D-GLDPC code ensembles; many of the following
definitions and notations also appear in \cite{flanagan09:IEEE-IT}.

\begin{figure}
\begin{center}\includegraphics[%
  width=\columnwidth,
  keepaspectratio]{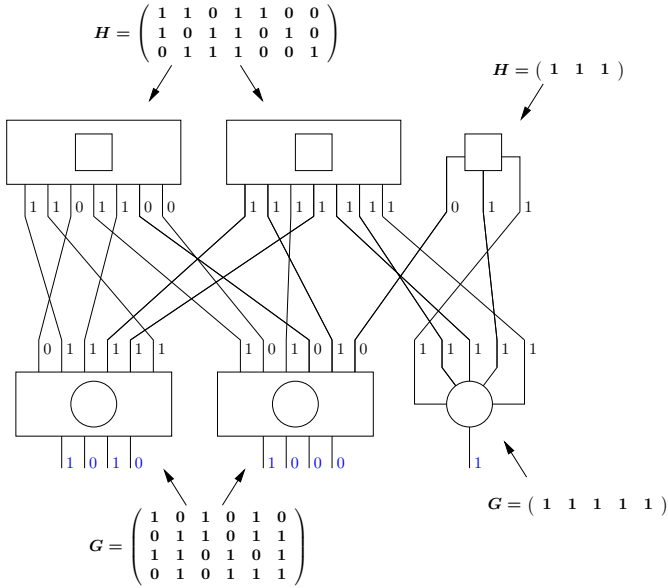}
  \end{center}
\caption{Illustration of an example D-GLDPC code, together with the D-GLDPC codeword $[1
\; 0 \; 1 \; 0 \; 1 \; 0 \; 0 \; 0 \; 1]$ (shown in blue on the diagram). Two of the
generalized CNs are $[7,4]$ Hamming codes, and the other is a length-$3$ SPC code. The VN set
consists of two $[6,4]$ linear codes and one length-$5$ repetition code (generator
matrices are shown). The reader may verify that the bits induced on the Tanner graph edges
through local encoding at the VNs correspond to valid local codewords from the CN perspective.
\label{cap:DGLDPC_example}}
\end{figure}

In the D-GLDPC code ensemble $\cM_n$, the number of VNs is denoted by $n$. The different CN types are denoted by the set $I_c = \{ 1,2,\ldots, n_c\}$; the local code $\Code_t$ of CN
type $t \in I_c$ has dimension, length and minimum distance denoted by $h_t$, $s_t$ and $r_t$,
respectively. Similarly, the different VN types are denoted by $I_v = \{ 1,2,\ldots, n_v\}$; the
local code $\Code_t$ of VN type $t \in I_v$ has dimension, length and minimum distance denoted by $k_t$, $q_t$
and $p_t$ respectively. We assume that the local codes associated with all VNs and CNs have minimum distance at least $2$ (i.e., $r_t \ge 2, p_t \ge 2$), and that the dual codes of all of these local codes have minimum distance greater than one\footnote{An equivalent condition is that the generator matrix of each code $\Code_t$ has no column consisting entirely of zeros.}.

For $t \in I_c$, the fraction of edges connected to CNs of type $t$ is denoted by $\rho_t$. Similarly, for $t \in I_v$, the fraction of edges connected to VNs of type $t$ is denoted by $\lambda_t$. CN and VN type-distribution polynomials are then given by $\rho(x)$ and $\lambda(x)$ respectively, where $\rho(x) \defeq \sum_{t\in I_c} \rho_t x^{s_t - 1}$ and $\lambda(x) \defeq \sum_{t \in I_v} \lambda_t x^{q_t - 1}$. If $E$ denotes the number of edges in the Tanner graph, the number of CNs of type $t\in I_c$ is then given by $E \rho_t / s_t$, and the number of VNs of type $t\in I_v$ is then given by $E \lambda_t / q_t$. Denoting as usual $\int_0^1 \rho(x) \, {\rm d} x$ and $\int_0^1 \lambda(x) \, {\rm d} x$ by $\int \rho$ and $\int \lambda$ respectively, the number of edges in the Tanner graph is given by $E = n / \int \lambda$ and the number of CNs is given by $m = E \int \rho$. Therefore, the fraction of CNs of type $t \in I_c$ and the fraction of VNs of type $t \in I_v$ are given by
\begin{equation}\label{eq:gamma_t_delta_t_definition}
\gamma_t = \frac{\rho_t}{s_t \int \rho} \quad \textrm{and} \quad \delta_t = \frac{\lambda_t}{q_t \int \lambda}
\end{equation}
respectively. A code in the irregular D-GLDPC ensemble corresponds to a permutation of the $E$ edges connecting VNs to CNs. The length of any D-GLDPC codeword in the ensemble is given by 
\begin{equation}
N = \sum_{t \in I_v} \left( \frac{E \lambda_t}{q_t} \right) k_t = \frac{n}{\int \lambda} \sum_{t \in I_v} \frac{\lambda_t k_t}{q_t} \; .
\label{eq:DG_LDPC_codeword_length}
\end{equation}
Note that this is a linear function of $n$. Similarly, the total number of parity-check equations
for any D-GLDPC code in the ensemble is given by $M = \frac{m}{\int \rho} \sum_{t \in I_c}
\frac{\rho_t (s_t - h_t)}{s_t}$.

The ensemble $\cM_n$ is defined according to a uniform probability distribution on all $E!$ permutations of the Tanner graph edges. The \emph{design rate} of the D-GLDPC ensemble is given by
\begin{equation}\label{eq:design_rate}
R = 1 - \frac{\sum_{t \in I_c} \rho_t (1 - R_t)}{\sum_{t \in I_v} \lambda_t R_t}
\end{equation}
where for $t \in I_c$ (resp. $t \in I_v$), $R_t$ is the local code rate of a type-$t$ CN (resp. VN). Each code in the ensemble has a code rate larger than or equal to $R$.

The WEF for CN type $t \in I_c$ is given by 
\begin{equation*}
A^{(t)}(z) = 1 + \sum_{u=r_t}^{s_t} A_u^{(t)} z^u \; . 
\end{equation*}
Here, for each $0 \le u \le s_t$, $A_u^{(t)} \ge 0$ denotes the number of weight-$u$ codewords for CNs of type $t$. We denote by $\bar{u}_t$ the maximal weight of a codeword in the local code for CN type $t$; this is the largest $u \in \{r_t,r_t+1,\dots,s_t\}$ such that $A^{(t)}_u > 0$. The input-output weight enumerating function (IO-WEF) for VN type $t \in I_v$ is given by
\begin{equation*}
B^{(t)}(x,y) = 1 + \sum_{u=1}^{k_t} \sum_{v=p_t}^{q_t} B_{u,v}^{(t)} x^u y^v \; .  
\end{equation*}
Here $B_{u,v}^{(t)} \ge 0$ denotes the number of weight-$v$ codewords generated by input words of weight $u$, for VNs of type $t$. Also, $B^{(t)}_{v}$ is the total number of weight-$v$ codewords for VNs of type $t$.

Although this paper is focused on the weight spectrum, all of the results developed in Sections~\ref{section:growth_rate}--\ref{section:spectral_shape_formula} can be extended to the stopping set size spectrum. A stopping set of a D-GLDPC code may be defined as any subset $\mathcal{S}$ of the code bits such that, assuming all code bits in $\mathcal{S}$ are erased and all code bits not in $\mathcal{S}$ are not erased, local erasure decoding at the CNs and VNs cannot recover any code bit in $\mathcal{S}$, so no erasure can be recovered by iterative decoding.\footnote{The concept of stopping set was first introduced in \cite{di02:finite} in the context of LDPC codes. When applied to LDPC codes (i.e., all CNs are SPC codes), the definition of stopping set used in this paper coincides with that in \cite{di02:finite}.} A \emph{local stopping set} for a CN is a subset of the local code bits which, if erased, is not recoverable to any extent by the CN. A \emph{local stopping set} for a VN is a subset of the local code bits together with a subset of the local information bits which, if both subsets are erased, are not recoverable to any extent by the VN. All results derived in this paper for the weight spectrum can be extended to the stopping set size spectrum by simply replacing the WEF for CN type $t \in I_c$ with its local stopping set enumerating function (SSEF), and replacing the IO-WEF for VN type $t \in I_v$ with its local input-output stopping set enumerating function (IO-SSEF).

We point out that for both VNs and CNs, the local SSEF depends on the decoding algorithm
used to locally recover from erasures; for local maximum \emph{a posteriori} (MAP) decoding, the reader is referred to
\cite[Appendix B]{flanagan09:IEEE-IT} for details.

Finally, we introduce some mathematical notation. Let $a(n)$ and $b(n)$ be two real-valued sequences, where $b(n) \ne 0$ for all $n$; we say that $a(n)$ is \emph{exponentially equivalent} to $b(n)$ as $n \rightarrow \infty$, writing $a(n) \asympequaln b(n)$, if and only if $\lim_{n \rightarrow \infty} n^{-1} \log \left( a(n) / b(n) \right) = 0$. Throughout this paper, the notation $e = \exp(1)$ denotes Napier's number, all the logarithms are assumed to have base $e$ and for $0<x<1$ the notation $h(x)=-x \log(x) - (1-x) \log(1-x)$ denotes the binary entropy function expressed in nats. 


\section{Weight Spectral Shape of Irregular D-GLDPC Code Ensembles}\label{section:growth_rate}
The weight spectral shape of the irregular D-GLDPC ensemble sequence $\{ \cM_n \}$ is defined by 
\begin{equation}
G(\alpha) \defeq \lim_{n\rightarrow \infty} \frac{1}{n} \log \mathbb{E}_{\cM_n} \left[ A_{\alpha n} \right]
\label{eq:growth_rate_result}
\end{equation}
where $\mathbb{E}_{\cM_n}$ denotes the expectation operator over the ensemble $\cM_n$, and $A_{w}$ denotes the number of codewords of weight $w$ of a randomly chosen D-GLDPC code in the ensemble $\cM_n$. The limit in (\ref{eq:growth_rate_result}) assumes the inclusion of only those positive integers $n$ for which $\alpha n \in \mathbb{Z}$ and $\mathbb{E}_{\cM_n} [ A_{\alpha n} ]$ is positive. Note that the argument of the growth rate function $G(\alpha)$ is equal to the ratio of D-GLDPC codeword weight to the number of VNs; by (\ref{eq:DG_LDPC_codeword_length}), this captures the behaviour of codewords whose weight is linear in the block length. 

Using standard notation \cite{orlitsky05:stopping}, we define the \emph{critical exponent codeword
weight ratio} for $\cM_n$ as $\alpha^* \defeq \inf \{ \alpha > 0 \; | \; G(\alpha)\geq 0 \}$. A
D-GLDPC ensemble is said to have \emph{good growth behavior} if $\alpha^* > 0$, and is said to have
\emph{bad growth rate behavior} if $\alpha^*=0$. In \cite{flanagan09:IEEE-IT}, it was shown that a
D-GLDPC ensemble always has good growth rate behavior if there exist no CNs or VNs with minimum
distance $2$ while, if there exist both CNs and VNs with minimum distance $2$, the ensemble has good
growth rate behavior if and only if $C \cdot V < 1$, where the (positive) parameters $C$
and $V$ are given by 
\begin{equation}
C = 2 \sum_{t \; : \; r_t = 2} \frac{\rho_t A^{(t)}_{2}}{s_t} \: ; \: V = 2 \sum_{t \; : \; p_t = 2} \frac{\lambda_t B^{(t)}_{2}}{q_t} \; .
\label{eq:C_V_definitions}
\end{equation}

Note that using (\ref{eq:DG_LDPC_codeword_length}), we may also define the growth rate with respect
to the D-GLDPC code's block length $N$ as follows:
\begin{equation}
H(\omega) \defeq \lim_{N\rightarrow \infty} \frac{1}{N} \log \mathbb{E}_{\cM_n} \left[ A_{\omega N} \right] \; .
\label{eq:growth_rate_result_norm}
\end{equation}
Note that we must have $\omega \le 1$, while in general we may have $\alpha > 1$. It is straightforward to show that 
\begin{equation}
H(\omega) = \frac{G(K_s \omega)}{K_s}
\label{eq:relationship_G_H}
\end{equation}
where $K_s$ is defined as the ratio of the D-GLDPC code's block length to the number of VNs, i.e.,
\begin{equation}
K_s \defeq \frac{N}{n} = \frac{1}{\int \lambda} \sum_{t \in I_v} \frac{\lambda_t k_t}{q_t} \; .
\label{eq:scaling_factor_def}
\end{equation}
Note that the parameter $K_s$ is independent of $N$.

The stopping set size spectral shapes of the ensemble sequence $\{ \cM_n \}$ for the case of bounded distance (BD) and MAP decoding at the CNs, whose definitions are analogous to \eqref{eq:growth_rate_result}, will be denoted by $G_{\Psi}(\alpha)$ and $G_{\Phi}(\alpha)$, respectively. Similarly the critical exponent stopping set size ratio will be denoted in these cases by $\alpha_{\Psi}^*$ and $\alpha_{\Phi}^*$, respectively.

In this section, we formulate an expression for the growth rate for an irregular D-GLDPC ensemble $\cM_n$ over a wider range of $\alpha$ than was considered in \cite{flanagan09:IEEE-IT} (where the case of small $\alpha$ was analyzed).

The following theorem constitutes the main result of this paper.
\medskip
\begin{theorem}
The weight spectral shape of the irregular D-GLDPC ensemble sequence $\{ \cM_n \}$ is given by
\begin{multline}
G(\alpha) = \sum_{t \in I_v} \delta_t \log B^{(t)}(x_{0},y_{0}) - \alpha \log x_{0} \\
+ \left( \frac{\int \rho}{\int \lambda} \right) \sum_{s \in I_c} \gamma_s \log A^{(s)}(z_0) + \frac{\log \left( 1 - \beta \int \lambda \right)}{\int \lambda} 
\label{eq:growth_rate_polynomial_general}
\end{multline}
where $x_0$, $y_0$, $z_0$ and $\beta$ are the unique positive real solutions to the $4 \times 4$ system of polynomial equations\footnote{Note that while~(\ref{eq:z0_eqn}),~(\ref{eq:x0_y0_eqn_1}) and~(\ref{eq:x0_y0_eqn_2}) are not polynomial as set down here, each may be made polynomial by multiplying across by an appropriate factor.}
\begin{equation}
\label{eq:z0_eqn}
z_0 \left( \frac{\int \rho}{\int \lambda} \right) \sum_{t \in I_c} \gamma_t \frac{ \frac{\mathrm{d} A^{(t)}}{\mathrm{d} z} (z_{0})}{A^{(t)}(z_{0})} = \beta \; ,
\end{equation}
\begin{equation}
\label{eq:x0_y0_eqn_1}
x_0 \sum_{t \in I_v} \delta_t \frac{ \frac{\partial B^{(t)}}{\partial x} (x_{0},y_{0})}{B^{(t)}(x_{0},y_{0})} = \alpha \; ,
\end{equation}
\begin{equation}
\label{eq:x0_y0_eqn_2}
y_0 \sum_{t \in I_v} \delta_t \frac{ \frac{\partial B^{(t)}}{\partial y} (x_{0},y_{0})}{B^{(t)}(x_{0},y_{0})} = \beta \; ,
\end{equation}
and
\begin{equation}
\label{eq:z0_y0_relation}
\left( \beta \smallint \lambda \right) (1 + y_0 z_0) = y_0 z_0 \; .
\end{equation}
\label{thm:growth_rate} 
\end{theorem}

\medskip
This theorem is proved in Section \ref{section:proof_of_main_result}. It is important to note that
the solution always involves a system of $4$ equations in $4$ unknowns, \emph{regardless} of the
number of different CN and VN types. Also, note that we can solve efficiently for the parameter
$\alpha^{*}$ without evaluating the entire spectral shape, by simply augmenting the
system~\eqref{eq:z0_eqn}--\eqref{eq:z0_y0_relation} with an additional equation which sets the
right-hand side of \eqref{eq:growth_rate_polynomial_general} to zero. Note also that substituting
for $\beta$ from \eqref{eq:z0_y0_relation} will further reduce the system to $3$ equations in $3$
unknowns. The reader may verify that in the special case of LDPC codes, this result reduces to
\cite[Corollary 12]{di06:weight}. 

We point out that Theorem~\ref{thm:growth_rate} is
consistent with Theorem~4.1 of \cite{flanagan09:IEEE-IT} which
provides an expression for $G(\alpha)$ valid for small $\alpha$. This can be seen by conducting an
analysis of \eqref{eq:growth_rate_polynomial_general}--\eqref{eq:z0_y0_relation} for the
small-$\alpha$ case. This analysis, detailed in Appendix~\ref{app:small_alpha},
yields the following slightly weaker version\footnote{The result is slightly weaker in the following sense: denoting the left-hand side of \eqref{eq:growth_rate_asymptotic_dgldpc1} by $F(\alpha) + o(\alpha)$, Corollary \ref{corollary:small_alpha} proves that $\lim_{n \rightarrow \infty} [G(\alpha) - F(\alpha)] / \alpha = 0$. In contrast, Theorem~4.1 of \cite{flanagan09:IEEE-IT} proves that $| G(\alpha) - F(\alpha)| \le K \cdot \alpha^{\xi}$, where $K$ is independent of $\alpha$ and $\xi > 0$ is a known parameter depending on the ensemble. Thus, Theorem~4.1 of \cite{flanagan09:IEEE-IT} provides a stronger statement regarding the rate of convergence of $G(\alpha)$ to $F(\alpha)$ as $\alpha \rightarrow 0$.} of Theorem~4.1 of \cite{flanagan09:IEEE-IT} as a corollary of Theorem~\ref{thm:growth_rate}.

\medskip
\begin{corollary}
\label{corollary:small_alpha}
The weight spectral shape of the irregular D-GLDPC ensemble sequence $\{ \cM_n \}$ is given by
\begin{multline}
G(\alpha) = \frac{T}{\psi} \alpha \log \alpha + \alpha \Bigg[ \log \frac{1}{Q_1^{-1}(1)} \\ +
\frac{T}{\psi} \log \frac{1}{Q_2(Q_1^{-1}(1))} \Bigg] + o(\alpha) \; ,
\label{eq:growth_rate_asymptotic_dgldpc1}
\end{multline}
where $r$ denotes the smallest minimum distance over all CN types, $\psi = r/(r-1)$, and $T$ denotes
the minimum of $(j-\psi)/i$ over all types $t \in I_v$ and all pairs $(i,j)$ such that
$B^{(t)}_{i,j} > 0$. The set $Y_v$ is the set of all types $t \in I_v$ such that this minimum $T$ is
achieved, and $P_t$ denotes the corresponding set of pairs $(i,j)$. Finally,
\begin{equation}
Q_1(x) = \sum_{t \in Y_v} \frac{\lambda_t}{q_t} \sum_{ (i,j) \in P_t } j B^{(t)}_{i,j} C^{j/r}
\left( \frac{\int \lambda}{e} \right)^{iT/\psi} x^i
\label{eq:P1x_definition}
\end{equation}
and
\begin{equation}
Q_2(x) = \sum_{t \in Y_v} \frac{\lambda_t}{q_t} \sum_{ (i,j) \in P_t } i B^{(t)}_{i,j} C^{j/r}
\left( \frac{\int \lambda}{e} \right)^{iT/\psi} x^i \; ,
\label{eq:P2x_definition}
\end{equation}
where 
\begin{equation}
\label{eq:general_C_definition1}
C = r \sum_{t \; : \; r_t = r} \frac{\rho_t A^{(t)}_{r}}{s_t} > 0
\end{equation} 
(this definition of the parameter $C$ reduces to that given in (\ref{eq:C_V_definitions}) in the
case $r=2$).
\end{corollary}

\medskip
When the term $(T/\psi)\alpha \log\alpha$ is nonzero, which happens when either $r\geq3$
or $p\geq3$ where $p$ is the smallest minimum distance over all VN types,
\eqref{eq:growth_rate_asymptotic_dgldpc1} may be directly used to obtain
an approximation to the parameter $\alpha^*$ without solving the polynomial system in
Theorem~\ref{thm:growth_rate} -- this approximation is in general very good for ensembles characterized by 
sufficiently small $\alpha^*$. To obtain this, we set $G(\alpha^*)=0$, $\alpha^* \ne 0$, and neglect the $o(\alpha)$ term in \eqref{eq:growth_rate_asymptotic_dgldpc1}, yielding
\begin{equation}\label{eq:alpha_star_approx_DGLDPC}
\alpha^* \approx \left[ Q_1^{-1}(1)\right]^{\frac{\psi}{T}} \cdot Q_2(Q_1^{-1}(1)) \; .
\end{equation}

For an irregular GLDPC code ensemble, this approximation reduces to
\begin{equation}\label{eq:alpha_star_approx_GLDPC}
\alpha^* \approx \lambda_p^{-r/(pr-p-r)} C^{-p/(pr-p-r)} \frac{e}{p \int\!\lambda}
\end{equation}
where $\lambda_p$ is the fraction of edges connected to the degree-$p$ VNs, i.e., to the VNs of
lowest degree, and where $C$ is defined in \eqref{eq:general_C_definition1}. If the GLDPC
code ensemble is variable-regular, \eqref{eq:alpha_star_approx_GLDPC} reduces to
\begin{equation}\label{eq:alpha_star_approx_GLDPC_variable_regular}
\alpha^* \approx C^{-p/(pr-p-r)} e\, .
\end{equation}

An even simpler expression is obtained for regular LDPC code ensembles of VN degree
$p\geq3$ (as $r=2$ for LDPC codes), for which \eqref{eq:alpha_star_approx_GLDPC_variable_regular}
becomes
\begin{equation}\label{eq:alpha_star_approx_LDPC_regular}
\alpha^* \approx \frac{e}{(d_c-1)^{1/(1-2/d_v)}}
\end{equation}
where $d_v=p$ is the VN degree and $d_c$ is the CN degree. Some numerical results on this
approximation will be presented in Section~\ref{section:examples}.

\medskip
\begin{lemma}
\label{lemma:Gprime_alpha}
The derivative of the weight spectral shape of the irregular D-GLDPC ensemble sequence $\{ \cM_n \}$ is given by
\begin{equation*}
G'(\alpha) = - \log x_{0} \; ,
\end{equation*}
where, for any $\alpha$, $x_0$ is given by the solution to the system of equations \eqref{eq:z0_eqn}--\eqref{eq:z0_y0_relation}. It follows that the stationary points of $G(\alpha)$ occur at exactly those values of $\alpha$ for which the solution to \eqref{eq:z0_eqn}--\eqref{eq:z0_y0_relation} satisfies $x_0 = 1$.
\end{lemma}
\medskip
\begin{proof} 
Note that in the solution for the weight spectral shape function given by Theorem \ref{thm:growth_rate}, each of the parameters $x_0$, $y_0$, $z_0$ and $\beta$ can be regarded as an implicit function of $\alpha$. Hence, differentiating \eqref{eq:growth_rate_polynomial_general} directly and using \eqref{eq:z0_eqn}, \eqref{eq:x0_y0_eqn_1} and \eqref{eq:x0_y0_eqn_2} yields
\begin{eqnarray*}
G'(\alpha) & = & - \log x_0 + \frac{\beta}{z_0} \frac{\mathrm{d} z_0}{\mathrm{d} \alpha} + \frac{\beta}{y_0} \frac{\mathrm{d} y_0}{\mathrm{d} \alpha} - \frac{1}{1-\beta \smallint \lambda} \frac{\mathrm{d} \beta}{\mathrm{d} \alpha} \\
& = & - \log x_0 + \beta \frac{\mathrm{d} \log \left( y_0 z_0 \right) }{\mathrm{d} \alpha} - \frac{y_0 z_0}{\beta \smallint \lambda} \frac{\mathrm{d} \beta}{\mathrm{d} \alpha} \; ,
\end{eqnarray*}
where we have used \eqref{eq:z0_y0_relation} in the final line. It remains to prove that the final pair of terms in this expression sum to zero. To show this, we write \eqref{eq:z0_y0_relation} as
\[
\log \left( \beta \smallint \lambda \right) + \log \left( 1 + y_0 z_0 \right) = \log \left( y_0 z_0 \right)
\] 
Differentiating this expression yields
\[
\frac{1}{\beta} \frac{\mathrm{d} \beta}{\mathrm{d} \alpha} = \frac{1}{y_0 z_0 (1 + y_0 z_0)} \frac{\mathrm{d} \left[ y_0 z_0 \right]}{\mathrm{d} \alpha} 
\]
and so
\begin{eqnarray*}
\frac{y_0 z_0}{\beta \smallint \lambda} \frac{\mathrm{d} \beta}{\mathrm{d} \alpha} & = & \frac{1}{(1 + y_0 z_0) \smallint \lambda} \frac{\mathrm{d} \left[ y_0 z_0 \right]}{\mathrm{d} \alpha} \\
& = & \frac{\beta}{y_0 z_0} \frac{\mathrm{d} \left[ y_0 z_0 \right]}{\mathrm{d} \alpha} = \beta \frac{\mathrm{d} \left[ \log (y_0 z_0) \right]}{\mathrm{d} \alpha} \; ,
\end{eqnarray*}
where in the final line we have again used \eqref{eq:z0_y0_relation}. This establishes the result of the Lemma.
\end{proof}

\medskip
\begin{theorem}
For any D-GLDPC code ensemble, the weight spectral shape $H(\omega)$ defined in \eqref{eq:relationship_G_H} has a stationary point at $\omega=1/2$, and the function value at this point is given by $H(1/2) = R \log 2$.
\label{thm:max_of_weight_spectral_shape} 
\end{theorem}

\medskip
\begin{proof}
Note that (using \eqref{eq:relationship_G_H}) it is equivalent to show that the weight spectral shape $G(\alpha)$ of Theorem \ref{thm:growth_rate} has a stationary point at $\alpha = K_s/2$, at which point $G(\alpha) = K_s R \log 2$. We will show that 
\begin{equation}
\alpha = \frac{K_s}{2} \; ; \; \quad x_0=y_0=z_0=1 \; ; \; \beta = \frac{1}{2 \int \lambda} \; , 
\label{eq:substitutions}
\end{equation}
makes \eqref{eq:z0_eqn}--\eqref{eq:z0_y0_relation} a consistent system of equations. This solution has the property that $x_0 = 1$, and therefore by Lemma \ref{lemma:Gprime_alpha}, a stationary point of the weight spectral shape exists at $G(K_s/2) = K_s R \log 2$ (or equivalently at $H(1/2) = R \log 2$).

First, it is easy to check that \eqref{eq:z0_y0_relation} holds under \eqref{eq:substitutions}.
Making the substitutions \eqref{eq:substitutions} in \eqref{eq:z0_eqn} yields
\begin{equation}
( \smallint\! \rho ) \sum_{t \in I_c} \gamma_t \frac{ \frac{\mathrm{d} A^{(t)}}{\mathrm{d} z} (1)}{A^{(t)}(1)} = \frac{1}{2} \; .
\label{eq:checkmax_z0_eqn_1}
\end{equation}
Note that $A^{(t)}(1) = \sum_{u=0}^{s_t} A_u^{(t)} = 2^{h_t}$ and $\frac{\mathrm{d} A^{(t)} } {\mathrm{d} z} (1) = \sum_{u=0}^{s_t} u A_u^{(t)} = \sum_{\bldcc \in \Code_t} w_{\mathrm{H}} (\bldc)$, where $w_{\mathrm{H}} (\bldc)$ denotes the Hamming weight of the codeword $\bldc$. Therefore, \eqref{eq:checkmax_z0_eqn_1} reduces to
\begin{equation}
( \smallint\! \rho ) \sum_{t \in I_c} \gamma_t \left[ \frac{ \sum_{\bldcc \in \Code_t} w_{\mathrm{H}} (\bldc) }{ 2^{h_t} } \right] = \frac{1}{2} \; .
\label{eq:checkmax_z0_eqn_2}
\end{equation}
Note that the quantity in square brackets is equal to the \emph{average weight} of a codeword in the local code $\Code_t$. Since we assume that the dual code of $\Code_t$ has minimum distance greater than one, it follows from the MacWilliams identities (see, e.g., \cite[Section 8.2]{pless:intro_to_theory_of_ecc}) that the average codeword weight is equal to $s_t/2$. Using this fact, \eqref{eq:checkmax_z0_eqn_2} reduces to
\[
( \smallint\! \rho ) \sum_{t \in I_c} \frac{\gamma_t s_t}{2} = \frac{1}{2} \; ,
\] 
which may be seen using the definition \eqref{eq:gamma_t_delta_t_definition} of $\gamma_t$ and the fact that $\sum_{t \in I_c} \rho_t = 1$.
Similarly, making the substitutions \eqref{eq:substitutions} in \eqref{eq:x0_y0_eqn_2} yields
\[
\sum_{t \in I_v} \delta_t \left[ \frac{ \sum_{\bldcc \in \Code_t} w_{\mathrm{H}} (\bldc) }{ 2^{k_t} } \right] = \frac{1}{2 \int \lambda} \; .
\]
Since we assume that all VNs have dual codes with minimum distance greater than one, this reduces to
\[
\sum_{t \in I_v} \frac{\delta_t q_t}{2} = \frac{1}{2 \int \lambda} \; ,
\]
which may be seen using the definition \eqref{eq:gamma_t_delta_t_definition} of $\delta_t$ and the fact that $\sum_{t \in I_v} \lambda_t = 1$.
Making the substitutions \eqref{eq:substitutions} in \eqref{eq:x0_y0_eqn_1} yields
\[
\sum_{t \in I_v} \delta_t \left[ \frac{ \sum_{u=0}^{k_t} \sum_{v=0}^{q_t} u B_{u,v}^{(t)} }{ 2^{k_t} } \right] = \frac{K_s}{2} \; .
\]
For $t \in I_v$, the quantity in square brackets is equal to the \emph{average weight of an information word} for the local code $\Code_t$; this is always equal to $k_t/2$. Thus we obtain
\[
\sum_{t \in I_v} \frac{\delta_t k_t}{2} = \frac{K_s}{2} \; ,
\]
which follows from the definition \eqref{eq:gamma_t_delta_t_definition} of $\delta_t$ and the definition \eqref{eq:scaling_factor_def} of $K_s$.
Finally, making the substitutions \eqref{eq:substitutions} in \eqref{eq:growth_rate_polynomial_general} yields
\begin{eqnarray*}
G(K_s/2) & = & \sum_{t \in I_v} \delta_t \log 2^{k_t} + \left( \frac{\int \rho}{\int \lambda} \right) \sum_{t \in I_c} \gamma_t \log 2^{h_t} - \frac{\log 2}{\int \lambda} \\
& = & \frac{\log 2}{\int \lambda} \left[ \sum_{t \in I_v} \frac{\lambda_t k_t}{q_t} + \sum_{t \in I_c} \frac{\rho_t h_t}{s_t} - 1 \right] \\
& = & \log 2 \left( \frac{\sum_{t \in I_v} \frac{\lambda_t k_t}{q_t}}{\int \lambda} \right) \left[ 1 - \frac{1 - \sum_{t \in I_c} \frac{\rho_t h_t}{s_t}}{\sum_{t \in I_v} \frac{\lambda_t k_t}{q_t}} \right] \\
& = & K_s R \log 2 \; ,
\end{eqnarray*}
where we have used the expression \eqref{eq:design_rate} for the (design) rate $R$ of the D-GLDPC code.
\end{proof}
\medskip
Although Theorem~\ref{thm:max_of_weight_spectral_shape} proves only that a stationary point of the weight spectral shape $H(\omega)$ exists at $\omega = 1/2$, we conjecture that this point represents a \emph{global maximum} of the weight spectral shape for \emph{any} D-GLDPC code ensemble. Indeed, this is empirically observed for all ensembles we have investigated, even though the total number of stationary points of the weight spectral shape is found to vary from 1 to 3 (c.f. Figures \ref{fig:growth_rate_rep2_H7}, \ref{fig:growth_rate_check_hybrid} and \ref{fig:growth_rate_asymp_bad} later in this paper). Previous work for the special case of LDPC codes also assumed this point to be a global maximum (c.f. the statement of Lemma~7 in \cite{measson08:maxwell_construction}).
 
Note that Theorem \ref{thm:max_of_weight_spectral_shape} only holds in general for the \emph{weight} spectral shape, and not for the \emph{stopping set size} spectral shape. It is interesting to note that this phenomenon, where the maximum weight spectral shape value of $R \log 2$ occurs at half the block length, appears to be quite a general one, occurring widely across many ensembles: for example, all of the spectral shape plots for protograph-based LDPC codes contained in \cite{abu-surra11:IEEE-IT} have this property (these were obtained by non-analytical optimization methods). Also, Gallager's ensemble described in \cite[Section 2.1]{gallager63:low-density}, where the parity-check matrix contains statistically independent equiprobable binary entries, shares the same property (note that this is not a low-density code ensemble).


\section{Spectral Shape of Check-Hybrid GLDPC Codes}\label{section:spectral_shape_formula}
In this section we consider the special case of a D-GLDPC code ensemble where all VNs are repetition
codes of the same length, i.e., a check-hybrid GLDPC code ensemble with regular VN set. The proofs of all lemmas in this section are deferred to Section \ref{section:supporting_proofs}.

\subsection{Evaluation of the Spectral Shape}

In the following, we show that a compact expression for the spectral shape follows in this case as a
natural corollary to Theorem~\ref{thm:growth_rate}. First we introduce the following definition (recall that for $t \in I_c$, $\bar{u}_t$ is the maximal weight of a codeword in the local code for CN type $t \in I_c$).

\medskip
\begin{definition} Let 
\begin{equation}
\mathsf{M} \defeq \left( \smallint\! \rho \right) \sum_{t \in I_c} \gamma_t \bar{u}_t \leq 1
\label{eq:M_definition}
\end{equation}
and define the function $\sff: \mathbb{R}^+ \rightarrow [0,\mathsf{M})$ as
\begin{align}\label{eq:f}
\sff(z) = \left( \smallint\! \rho \right) \sum_{t \in I_c} \gamma_t \frac{ z \, \frac{\mathrm{d}
A^{(t)}(z)}{\mathrm{d} z}}{A^{(t)}(z)} \, .
\end{align}
\end{definition}

\medskip
Note that we have $\mathsf{M}=1$ if and only if $\bu_t=s_t$ for all $t \in I_c$.

\medskip
\begin{lemma}\label{lemma:f_properties} 
The function $\sff$ fulfills the following properties:
\begin{enumerate}
\item $\sff(0) = \sff'(0) = 0$;
\item $\sff$ is monotonically increasing for all $z>0$;
\item $\lim_{z\rightarrow +\infty} \sff(z) = \mathsf{M}$.
\end{enumerate}
\end{lemma}

\medskip
Note that, due to Lemma \ref{lemma:f_properties}, the inverse of $\sff$, denoted by
$\sff^{-1}:[0,\mathsf{M})\rightarrow\mathbb{R}^+$, is well-defined. We are now in a position to
state the main result of this section.

\medskip
\begin{theorem}
Consider a GLDPC code ensemble with a regular VN set, composed of repetition codes all of length
$q$, and a hybrid CN set, composed of a mixture of $n_c$ different linear block code types. Then,
the weight spectral shape of the ensemble is given~by
\begin{align}\label{eq:G(alpha)_irregular_CN_set}
G(\alpha) & = (1-q) h(\alpha) -q\, \alpha \log \sff^{-1}(\alpha) \notag \\ 
\, & \phantom{------}+ q \left(\smallint\!\rho\right) \sum_{t\in I_c} \gamma_t \log A^{(t)}(\sff^{-1}(\alpha)) \; .
\end{align}
%
\label{thm:main_result}
\end{theorem}
\begin{proof} 
The ensemble constitutes a special case of a D-GLDPC code ensemble where all VNs are repetition codes of length $q \geq 2$, with corresponding IO-WEF
\begin{align}\label{eq:B_repetition}
B(x,y)=1+xy^q \, .
\end{align}
Using Theorem \ref{thm:growth_rate}, the spectral shape function simplifies to (noting that $\int \lambda = 1/q$ in this case)
\begin{align}
G(\alpha) & = \log B(x_{0},y_{0}) - \alpha \log x_{0} \notag \\ 
\, & \phantom{--} + q \left(\smallint\! \rho \right) \sum_{t \in I_c} \gamma_t \log A^{(t)}(z_0) + q \log \left( 1 - \frac{\beta}{q} \right)
\label{eq:growth_rate_polynomial_Tanner}
\end{align}
where the values of $x_0$, $y_0$, $z_0$, $\beta$ in \eqref{eq:growth_rate_polynomial_Tanner} are found by solving the $4 \times 4$ polynomial system
\begin{equation} 
\left( \smallint\! \rho \right) \sum_{t \in I_c} \gamma_t \frac{z_0 \, \frac{\mathrm{d} A^{(t)}(z_0)}{\mathrm{d} z}}{A^{(t)}(z_0)} = \frac{\beta}{q} \; ,
\label{eq:Tanner_eq1}
\end{equation}
\begin{equation} 
\frac{x_0 y_0^q}{1 + x_0 y_0^q} = \alpha \; ,
\label{eq:Tanner_eq2}
\end{equation}
\begin{equation} 
\frac{x_0 y_0^q}{1 + x_0 y_0^q} = \frac{\beta}{q} \; ,
\label{eq:Tanner_eq3}
\end{equation}
and
\begin{equation} 
\frac{z_0 y_0}{1 + z_0 y_0} = \frac{\beta}{q} \; .
\label{eq:Tanner_eq4}
\end{equation}
Note that we are certain of the existence of a unique real solution to the polynomial system such that $x_0>0$, $y_0>0$, $z_0>0$, $\beta>0$, due to Hayman's formula. We solve this system of equations sequentially for the variables $\beta$, $z_0$, $y_0$ and $x_0$ (respectively). First, combining (\ref{eq:Tanner_eq2}) and (\ref{eq:Tanner_eq3}) yields 
\begin{equation} 
\beta = q \alpha \; .
\label{eq:Tanner_eq_beta}
\end{equation}
Substituting (\ref{eq:Tanner_eq_beta}) into (\ref{eq:Tanner_eq1}) yields $\sff(z_0) = \alpha$ which may be written as 
\begin{equation} 
z_0 = \sff^{-1} (\alpha) \; .
\label{eq:Tanner_eq_z0}
\end{equation}
Using (\ref{eq:Tanner_eq_beta}) and (\ref{eq:Tanner_eq_z0}) in (\ref{eq:Tanner_eq4}) yields 
\begin{equation} 
y_0 = \frac{\alpha}{(1-\alpha) \sff^{-1}(\alpha)} \; .
\label{eq:Tanner_eq_y0}
\end{equation}
Finally, substituting (\ref{eq:Tanner_eq_beta}) and (\ref{eq:Tanner_eq_y0}) into (\ref{eq:Tanner_eq3}) yields
\begin{equation} 
x_0 = \left( \frac{\alpha}{1-\alpha} \right)^{1-q} \left( \sff^{-1} (\alpha) \right)^q \; .
\label{eq:Tanner_eq_x0}
\end{equation}
Substituting (\ref{eq:Tanner_eq_beta}), (\ref{eq:Tanner_eq_z0}), (\ref{eq:Tanner_eq_y0}) and (\ref{eq:Tanner_eq_x0}) into (\ref{eq:growth_rate_polynomial_Tanner}), and simplifying, leads to \eqref{eq:G(alpha)_irregular_CN_set}. 
\end{proof}

\medskip
The expression \eqref{eq:G(alpha)_irregular_CN_set} holds regardless of whether the ensemble has
good or bad growth rate behavior. Note that, according to \eqref{eq:G(alpha)_irregular_CN_set},
the growth rate $G(\alpha)$ is well-defined only for $\alpha \in [0,\mathsf{M}]$. This is as
expected due to
the following reasoning. A codeword of weight $\alpha n$ naturally induces a distribution of bits on
the Tanner graph edges, $\alpha n q$ of which are equal to $1$. Also note that the maximum number of
ones in this distribution occurs when a maximum weight local codeword is activated for each of the
$\gamma_t m$ CNs of type $t \in I_c$, and is thus given by $m \sum_{t \in I_c} \gamma_t \bar{u}_t$.
Hence, we have $\alpha n q \leq m \sum_{t \in I_c} \gamma_t \bar{u}_t$, i.e., $\alpha \leq
\mathsf{M}$.

By considering Theorem \ref{thm:main_result} in the special case of Tanner codes, we obtain the following corollary.

\medskip
\begin{corollary}\label{theorem:G(alpha)_tanner_codes}
Consider a Tanner code ensemble where all variable component codes are length-$q$ repetition codes and where all check component codes are length-$s$ codes with weight enumerating function $A(z)=1+\sum_{u=r}^s A_u z^u$. The weight spectral shape of this ensemble is given by
\begin{equation}\label{eq:G(alpha)_tanner_codes}
G(\alpha) = (1-q) h(\alpha) - q\,\alpha\,\log(\sff^{-1}(\alpha)) + \frac{q}{s}\log A(\sff^{-1}(\alpha)) 
\end{equation}
\noindent where the function $\sff$ is given by (special case of (\ref{eq:f}))
\begin{align}\label{eq:f_regular_case}
\sff(z) = \frac{z\, A'(z)}{s\, A(z)} \, ,
\end{align}
and $\sff^{-1}:[0,\mathsf{M})\rightarrow\mathbb{R}^+$ is well-defined, where
$\mathsf{M}=\frac{\bu}{s}$ and $\bar{u}$ denotes the largest $u \in \{r,r+1,\dots,s\}$ such that
$A_u > 0$.
\end{corollary}

\medskip
Note that, in the special case where all CNs are SPC codes, \eqref{eq:G(alpha)_tanner_codes} becomes equal to the spectral shape expression for regular LDPC codes developed in \cite[Theorem~2]{orlitsky05:stopping} for the case of stopping sets. Also note that, in some cases, \eqref{eq:G(alpha)_tanner_codes} can be expressed analytically as $\sff^{-1}(\alpha)$ admits an analytical form. As shown in Appendix~\ref{appendix:closed_form} this is the case, for instance, of $(3,6)$ and $(4,8)$ regular LDPC code ensembles. This shows that some of Gallager's $B_{j,k}(\lambda)$ functions \cite{gallager63:low-density} can be expressed in closed form.

\subsection{Symmetry of the Weight Spectral Shape}\label{section:symmetry}
As in the previous subsection, consider a GLDPC code ensemble with a regular VN set and a hybrid CN
set. Next, we show how a symmetry in the overall weight spectral shape of the ensemble is
induced by local symmetry properties in the WEFs of the CNs.

%

\begin{definition}\label{def:A_symmetry}
The WEF of CN type $t \in I_c$ is said to be \emph{symmetric} if and only if
$A^{(t)}_{\bar{u}_t-u}=A^{(t)}_u$ for all $u \in \{ 0, 1, \ldots, \bar{u}_t \}$.
\end{definition}
%
%
\medskip
\begin{lemma}\label{lemma:all_1_codeword}
The WEF of CN type $t \in I_c$ is symmetric if and only if the all-$1$ codeword belongs to the code.
\end{lemma}

\medskip
Lemma~\ref{lemma:all_1_codeword} proves that the WEF of a linear block code is symmetric if and only if
$\bar{u}_t=s_t$.

\medskip
\begin{lemma}\label{lemma:A_symmetry}
The WEF of CN type $t \in I_c$ fulfills
\begin{equation}
A^{(t)}(z) = z^{\bar{u}_t} A^{(t)}\left(z^{-1}\right)
\label{eq:symmetry_of_CN_type} 
\end{equation} 
for all $z \in \mathbb{R}^+$ if and only if it is symmetric (equivalently, if and only if
$\bar{u}_t=s_t$).
\end{lemma}


\medskip
\begin{lemma}\label{lemma:finv_symmetry}
If $A^{(t)}(z)$ is symmetric for every \mbox{$t \in I_c$}
(i.e., if $\mathsf{M}=1$), then the inverse function $\sff^{-1}$ fulfills
\begin{align}\label{eq:finv_symmetry}
\sff^{-1}(\mathsf{M}-\alpha)=\frac{1}{\sff^{-1}(\alpha)}
\end{align}
$\forall$ $\alpha \in (0,M)$.
\end{lemma}

\medskip
\begin{theorem}\label{theorem:G_symmetry_sufficient} Consider a GLDPC code ensemble with a regular
VN set, composed of repetition codes all of length $q$, and a hybrid CN set, composed of a mixture
of $n_c$ different linear block code types. If $A^{(t)}(z)$ is symmetric for each $t \in I_c$
(equivalently, if $\mathsf{M}=1$), then the spectral shape of the ensemble fulfills
\begin{align}\label{eq:G_symmetry}
G(\mathsf{M}-\alpha) = G(\alpha)
\end{align}
for all $\alpha \in (0,\mathsf{M})$.
\end{theorem}
\begin{proof}
Assume that $A^{(t)}(z)$ is symmetric for each $t \in I_c$ (i.e., $\mathsf{M}=1$). From
\eqref{eq:G(alpha)_irregular_CN_set} we have:
\begin{align*}
\, & G(\mathsf{M}-\alpha) = (1-q) h(\mathsf{M}-\alpha) - q(\mathsf{M}-\alpha) \log
\sff^{-1}(\mathsf{M}-\alpha) \\ 
\, & \phantom{--------} + q \left(\smallint\! \rho\right) \sum_{t \in I_c} \gamma_t \log A^{(t)}
\left(\sff^{-1}(\mathsf{M}-\alpha)\right) \\
\, & \stackrel{\textrm{(a)}}{=} (1-q) h(\mathsf{M}-\alpha) - q(\mathsf{M}-\alpha) \log
\frac{1}{\sff^{-1}(\alpha)} \\ 
\, & \phantom{--------} + q \left(\smallint\! \rho\right) \sum_{t \in I_c} \gamma_t \log A^{(t)}
\left(\frac{1}{\sff^{-1}(\alpha)}\right) \\
\, & \stackrel{\textrm{(b)}}{=} (1-q) h(\mathsf{M}-\alpha) - q\, \alpha\, \log(\sff^{-1}(\alpha))
\\ 
\, & \phantom{----------} + q \left(\smallint\! \rho\right) \sum_{t \in I_c} \gamma_t \log A^{(t)}
(\sff^{-1}(\alpha)) \\
\, & \stackrel{\textrm{(c)}}{=} G(\alpha)
\end{align*}
where (a) follows from Lemma~\ref{lemma:finv_symmetry}, (b) from
Lemma~\ref{lemma:A_symmetry} and \eqref{eq:M_definition}, and (c) from
$\mathsf{M}=1$.
\end{proof}

\medskip
We remark that the converse of this result, i.e., that if $G(\alpha) = G(\mathsf{M}-\alpha)$ for all $\alpha \in (0,\mathsf{M})$, then $A^{(t)}(z)$ is symmetric for every \mbox{$t \in I_c$} (and therefore $\mathsf{M}=1$), appears to hold for almost all code ensembles; however this converse appears to be difficult to prove in the general case.

\section{Proof of Theorem~\ref{thm:growth_rate}}\label{section:proof_of_main_result}
In this section we prove Theorem \ref{thm:growth_rate}. The proof uses the concepts of
\emph{assignment} and \emph{split assignment}, defined next. These concepts were introduced in
\cite{di06:weight} and \cite{flanagan09:IEEE-IT}, respectively. 
\medskip
\begin{definition}
An \emph{assignment} is a subset of the edges of the Tanner graph. An assignment is said to have
\emph{weight} $k$ if it has $k$ elements. An assignment is said to be \emph{check-valid} if the
following condition holds: supposing that each edge of the assignment carries a $1$ and each of the
other edges carries a $0$, each CN recognizes a valid local codeword.
\end{definition}
\medskip
\begin{definition}
A \emph{split assignment} is an assignment, together with a subset of the D-GLDPC code bits (called
a \emph{codeword assignment}). A split assignment is said to have \emph{split weight} $(u, v)$ if
its assignment has weight $v$ and its codeword assignment has $u$ elements. A split assignment is
said to be \emph{check-valid} if its assignment is check-valid. A split assignment is said to be
\emph{variable-valid} if the following condition holds: supposing that each edge of its assignment
carries a $1$ and each of the other edges carries a $0$, and supposing that each D-GLDPC code bit in
the codeword assigment is set to $1$ and each of the other code bits is set to $0$, each VN
recognizes a local input word and the corresponding valid local codeword.     
\end{definition}
\medskip
For ease of presentation, the proof is broken into two parts.

\subsection{Number of Check-Valid Assignments of Weight $\delta
m$}\label{subsec:no_of_check_valid_assignments}

First we derive an expression, valid asymptotically, for the number of check-valid assignments of
weight $\delta m$. For each $t \in I_c$, let $\epsilon_t m$ denote the portion of the total weight
$\delta m$ apportioned to CNs of type $t$. Then $\epsilon_t \ge 0$ for each $t \in I_c$, and
$\sum_{t \in I_c} \epsilon_t = \delta$. Also denote $\bldepsilon = (\epsilon_1 \; \epsilon_2 \;
\cdots \; \epsilon_{n_c})$. 

Consider the set of $\gamma_t m$ CNs of a particular type $t \in I_c$, where $\gamma_t$ is given by
(\ref{eq:gamma_t_delta_t_definition}). Using generating functions, the number of check-valid
assignments (over these CNs) of weight $\epsilon_t m$ is given by
\[
N_{c,t}^{(\gamma_t m)}(\epsilon_t m) = \coeff \left[ \left( A^{(t)}(x) \right) ^{\gamma_t m},
x^{\epsilon_t m} \right]
\]
where $\coeff [ p(x), x^c ]$ denotes the coefficient of $x^c$ in the polynomial $p(x)$. We next
apply Lemma \ref{lemma:optimization_1D}, substituting $A(x) = A^{(t)}(x)$, $\ell=\gamma_t m$ and
$\xi = \epsilon_t/\gamma_t$; we obtain that as $m \rightarrow \infty$
\begin{eqnarray}
N_{c,t}^{(\gamma_t m)}(\epsilon_t m) = \coeff \left[ \left( A^{(t)}(x) \right) ^{\gamma_t m},
x^{\epsilon_t m} \right] 
\label{eq:Nct_epsilon_start} \\
\asympequalm \exp \left\{ m \left( \gamma_t \log A^{(t)}(z_{0,t}) - \epsilon_t \log z_{0,t} \right)
\right\} 
\label{eq:Nct_epsilon_mid}
\end{eqnarray}
where, for each $t \in I_c$, $z_{0,t}$ is the unique positive real solution to
\begin{equation}
\label{eq:z0t_soln_to_At_eqn}
\gamma_t \frac{ \frac{\mathrm{d} A^{(t)}}{\mathrm{d} z} (z_{0,t})}{A^{(t)}(z_{0,t})} \cdot z_{0,t} =
\epsilon_t \; .
\end{equation}

The number of check-valid assignments of weight $\delta m$ satisfying the constraint $\bldepsilon$
is obtained by multiplying the numbers of check-valid assignments of weight $\epsilon_t m$ over
$\gamma_t m$ CNs of type $t$, for each $t \in I_c$,
\begin{equation}
N_c^{(\bldepsilon)}(\delta m) = \prod_{t \in I_c} N_{c,t}^{(\gamma_t m)}(\epsilon_t m) \; .
\label{eq:Nc_epsilon}
\end{equation}
The number $N_c(\delta m)$ of check-valid assignments of weight $\delta m$ is then equal to the sum
of $N_c^{(\bldepsilon)}(\delta m)$ over all admissible vectors $\bldepsilon$; therefore
by~(\ref{eq:Nct_epsilon_mid}), as $m\rightarrow \infty$
\begin{equation}
N_c(\delta m) \, \asympequalm \sum_{\bldsmallepsilon \; : \; \sum_{t \in I_c} \epsilon_t = \delta}
\exp \left\{ m W(\bldepsilon) \right\}
\label{eq:sum_of_exp_check}
\end{equation}
where
\begin{equation}
W(\bldepsilon) = \sum_{t \in I_c} \left( \gamma_t \log A^{(t)}(z_{0,t}) - \epsilon_t \log z_{0,t}
\right) \; .
\label{eq:W_definition}
\end{equation}
As $m \rightarrow \infty$, the asymptotic expression is dominated by the distribution $\bldepsilon$
which maximizes the argument of the exponential function\footnote{Observe that as $m\rightarrow
\infty$, $\sum_t \exp ( m Z_t ) \asympequalm \exp ( m \max_t \{Z_t\} )$}. Therefore as $m\rightarrow
\infty$
\begin{equation}
N_c(\delta m) \asympequalm \exp \left\{ m X \right\}
\label{eq:N_c_as_function_of_X}
\end{equation}
where 
\begin{equation}
X = \max_{\bldsmallepsilon} W(\bldepsilon)
\end{equation} 
and the maximization is subject to the constraint
\begin{equation}
V(\bldepsilon) = \sum_{t \in I_c} \epsilon_t = \delta 
\label{eq:sum_epsilont_constraint}
\end{equation}
together with $\epsilon_t \ge 0$ for each $t \in I_c$, and for every $t \in I_c$, $z_{0,t}$ is the
unique positive real solution to~(\ref{eq:z0t_soln_to_At_eqn}). Note that for each $t \in I_c$,
(\ref{eq:z0t_soln_to_At_eqn}) provides an implicit definition of $z_{0,t}$ as a function of
$\epsilon_t$. 

We solve this optimization problem using Lagrange multipliers; at the maximum, we must have
\begin{equation}
\frac{\partial W(\bldepsilon)}{\partial \epsilon_t} = \lambda \frac{\partial
V(\bldepsilon)}{\partial \epsilon_t}
\label{eq:max_condition_check}
\end{equation}
for all $t \in I_c$, where $\lambda$ is the Lagrange multiplier. This yields
\begin{equation}
\frac{\partial z_{0,t}}{\partial \epsilon_t} \left[ \gamma_t \frac{ \frac{\mathrm{d}
A^{(t)}}{\mathrm{d} z} (z_{0,t})}{A^{(t)}(z_{0,t})} - \frac{\epsilon_t}{z_{0,t}} \right] - \log
z_{0,t} = \lambda \; .
\end{equation}
The term in square brackets is equal to zero due to (\ref{eq:z0t_soln_to_At_eqn}); therefore this
simplifies to $\log z_{0,t} = -\lambda$ for all $t \in I_c$.
We conclude that all of the $\{ z_{0,t} \}$ are equal, and we may write 
\begin{equation}
z_{0,t} = z_0 \quad \forall t \in I_c \; . 
\label{eq:z0_all_equal}
\end{equation}
Making this substitution in~(\ref{eq:N_c_as_function_of_X}) and
using~(\ref{eq:sum_epsilont_constraint}) we obtain
\begin{equation}
N_c(\delta m) \asympequalm \exp \left\{ m \left( \sum_{t \in I_c} \gamma_t \log A^{(t)}(z_{0}) -
\delta \log z_{0} \right) \right\} \; .
\label{eq:N_c_asymptotic}
\end{equation}
Summing~(\ref{eq:z0t_soln_to_At_eqn}) over $t \in I_c$ and using (\ref{eq:sum_epsilont_constraint})
and (\ref{eq:z0_all_equal}) implies that the value of $z_0$ in~(\ref{eq:N_c_asymptotic}) is the
unique positive real solution to~(\ref{eq:z0_eqn}) (here we have defined $\beta$ through the relationship 
$\beta n = \delta m$, and we have also used the fact that $n \int \rho = m \int \lambda$).

\subsection{Polynomial-System Solution for the Growth Rate}\label{subsec:growth_rate}
Consider the set of $\delta_t n$ VNs of a particular type $t \in I_v$, where $\delta_t$ is given by
(\ref{eq:gamma_t_delta_t_definition}). Using generating functions, the number of variable-valid
split assignments (over these VNs) of split weight $(\alpha_t n, \beta_t n)$ is given by
\[
N_{v,t}^{(\delta_t n)}(\alpha_t n, \beta_t n) = \coeff  \left[ \left( B^{(t)}(x,y) \right)
^{\delta_t n}, x^{\alpha_t n} y^{\beta_t n} \right]
\]
where $\coeff [p(x,y), x^c y^d ]$ denotes the coefficient of $x^c y^d$ in the bivariate polynomial
$p(x,y)$. Next we apply Lemma \ref{lemma:optimization_2D}, substituting $B(x,y) = B^{(t)}(x,y)$,
$\ell = \delta_t n$, $\xi = \alpha_t/\delta_t$ and $\theta = \beta_t/\delta_t$; we obtain that as $n
\rightarrow \infty$
\begin{eqnarray}
N_{v,t}^{(\delta_t n)}(\alpha_t n, \beta_t n) = \coeff  \left[ \left( B^{(t)}(x,y) \right)
^{\delta_t n}, x^{\alpha_t n} y^{\beta_t n} \right] \nonumber \\
\asympequaln \exp \left\{ n X_t^{(\delta_t)}(\alpha_t, \beta_t) \right\}
\label{eq:Nvt_asymptotic}
\end{eqnarray}
where 
\begin{equation}
X_t^{(\delta_t)}(\alpha_t, \beta_t) = \delta_t \log B^{(t)}(x_{0,t},y_{0,t}) - \alpha_t \log x_{0,t}
- \beta_t \log y_{0,t} 
\end{equation}
and where $x_{0,t}$ and $y_{0,t}$ are the unique positive real solutions to the pair of simultaneous
equations
\begin{equation}
\label{eq:x0t_y0t_eqn_1}
\delta_t \frac{ \frac{\partial B^{(t)}}{\partial x} (x_{0,t},y_{0,t})}{B^{(t)}(x_{0,t},y_{0,t})}
\cdot x_{0,t} = \alpha_t
\end{equation}
and
\begin{equation}
\label{eq:x0t_y0t_eqn_2}
\delta_t \frac{ \frac{\partial B^{(t)}}{\partial y} (x_{0,t},y_{0,t})}{B^{(t)}(x_{0,t},y_{0,t})}
\cdot y_{0,t} = \beta_t \; . 
\end{equation}

Next, note that the expected number of D-GLDPC codewords of weight $\alpha n$ in the ensemble
$\cM_n$ is equal to the sum over $\beta$ of the expected numbers of split assignments of split
weight $(\alpha n, \beta n)$ which are both check-valid and variable-valid, denoted $N^{v,c}_{\alpha
n, \beta n}$:
\[
\mathbb{E}_{\cM_n} \left[ A_{\alpha n} \right] = \sum_{\beta} \mathbb{E}_{\cM_n} [ N^{v,c}_{\alpha
n, \beta n} ] \; .
\]
This may then be expressed as
\begin{multline}
\mathbb{E}_{\cM_n} \left[ A_{\alpha n} \right] = \sum_{\substack{\alpha_t \ge 0, t \in I_v \\ \sum_t
\alpha_t = \alpha}} \sum_{\beta_t \ge 0, t \in I_v} P_{\mbox{\scriptsize c-valid}}(\beta n) \\
\times \prod_{t \in I_v} N_{v,t}^{(\delta_t n)}(\alpha_t n, \beta_t n)
\end{multline}
where 
\begin{equation}
\beta \defeq \sum_{t \in I_v} \beta_t \; .
\label{eq:beta_sum_of_betat}
\end{equation}
Here $P_{\mbox{\scriptsize c-valid}}(\beta n)$ denotes the probability that a randomly chosen
assignment of weight $\beta n$ is check-valid, and is given by
\[
P_{\mbox{\scriptsize c-valid}}(\beta n) = N_c(\beta n) \Big/ \binom{E}{\beta n} \; .
\]
Applying \cite[eqn. (25)]{di06:weight}, we find that as $n \rightarrow \infty$
\[
\binom{E}{\beta n} = \binom{n/\int \lambda}{\beta n} \asympequaln \exp  \left\{
\frac{n}{\int\lambda} h \left( \beta \smallint \lambda \right) \right\} \; .
\]
Combining this result with~(\ref{eq:N_c_asymptotic}), we obtain that as $n \rightarrow \infty$
\[
P_{\mbox{\scriptsize c-valid}}(\beta n) \asympequaln \exp \left\{ n Y(\beta) \right\} 
\]
where 
\[
Y(\beta) = \left( \frac{\int \rho}{\int \lambda} \right) \sum_{t \in I_c} \gamma_t \log \left(
A^{(t)}(z_0) \right) - \beta \log z_0 - \frac{h(\beta \int \lambda)}{\int \lambda} \; .
\] 
Therefore, as $n \rightarrow \infty$
\begin{align}\label{eq:G_as_sum_of_exp}
\, & \mathbb{E}_{\cM_n} \left[ A_{\alpha n} \right] \asympequaln \nonumber \\
\, & \sum_{\substack{\alpha_t \ge 0, t \in I_v \\ \sum_t \alpha_t = \alpha}} 
\sum_{\beta_t, t \in I_v} \exp \left\{ n \left( \sum_{t \in I_v} X_t^{(\delta_t)}(\alpha_t, \beta_t)
+ Y(\beta) \right) \right\} \; .
\end{align}
Note that the sum in (\ref{eq:G_as_sum_of_exp}) is dominated asymptotically by the term which
maximizes the argument of the exponential function. Thus, denoting the two vectors of independent
variables by $\bldalpha = (\alpha_t)_{t \in I_v}$ and $\bldbeta = (\beta_t)_{t \in I_v}$, we have 
\begin{eqnarray}
G(\alpha) = \max_{\bldsmallalpha, \bldsmallbeta} S(\bldalpha, \bldbeta)
\label{eq:numerical_evaluation_1_weak}
\end{eqnarray}
where
\begin{equation}
S(\bldalpha, \bldbeta) = \sum_{t \in I_v} X_t^{(\delta_t)}(\alpha_t, \beta_t) + Y(\beta) 
\label{eq:S_definition}
\end{equation}
where $\beta$ is given by (\ref{eq:beta_sum_of_betat}), and the maximization is subject to the
constraint 
\begin{equation}
R(\bldalpha, \bldbeta) = \sum_{t \in I_v} \alpha_t = \alpha 
\label{eq:sum_alpha_constraint}
\end{equation} 
together with $\alpha_t \ge 0$ and appropriate inequality constraints on $\beta_t$ for each $t \in
I_v$.

Note that (\ref{eq:z0_eqn}) provides an implicit definition of $z_0$ as a function of $\bldbeta$.
Similarly, for any $t \in I_v$, (\ref{eq:x0t_y0t_eqn_1}) and (\ref{eq:x0t_y0t_eqn_2}) provide
implicit definitions of $x_{0,t}$ and $y_{0,t}$ as functions of the two variables $\alpha_t$ and
$\beta_t$. 

We solve the constrained optimization problem using Lagrange multipliers; at the maximum, we must
have
\[
\frac{\partial S(\bldalpha, \bldbeta)}{\partial \alpha_t} = \mu \frac{\partial R(\bldalpha,
\bldbeta)}{\partial \alpha_t}  
\]
for all $t \in I_v$, where $\mu$ is the Lagrange multiplier. This yields
\begin{eqnarray*}
\frac{\partial x_{0,t}}{\partial \alpha_t} \left[ \delta_t \frac{ \frac{\partial B^{(t)}}{\partial
x} (x_{0,t},y_{0,t})}{B^{(t)}(x_{0,t},y_{0,t})} - \frac{\alpha_t}{x_{0,t}} \right] - \log x_{0,t} \\
+ \frac{\partial y_{0,t}}{\partial \alpha_t} \left[ \delta_t \frac{ \frac{\partial B^{(t)}}{\partial
y} (x_{0,t},y_{0,t})}{B^{(t)}(x_{0,t},y_{0,t})} - \frac{\beta_t}{y_{0,t}} \right] = \mu \; .
\end{eqnarray*}
The terms in square brackets are zero due to (\ref{eq:x0t_y0t_eqn_1}) and (\ref{eq:x0t_y0t_eqn_2})
respectively; therefore this simplifies to $\log x_{0,t} = -\mu$ for all $t \in I_v$.
We conclude that all of the $\{ x_{0,t} \}$ are equal, and we may write 
\begin{equation}
x_{0,t} = x_0 \quad \forall t \in I_v \; .
\label{eq:x0_all_equal}
\end{equation} 
At the maximum, we must also have
\[
\frac{\partial S(\bldalpha, \bldbeta)}{\partial \beta_t} = \mu \frac{\partial R(\bldalpha,
\bldbeta)}{\partial \beta_t}  
\]
for all $t \in I_v$. This yields
\begin{multline}
\frac{\partial x_{0,t}}{\partial \alpha_t} \left[ \delta_t \frac{ \frac{\partial B^{(t)}}{\partial
x} (x_{0,t},y_{0,t})}{B^{(t)}(x_{0,t},y_{0,t})} - \frac{\alpha_t}{x_{0,t}} \right] - \log y_{0,t} -
\log z_0 \\
+ \frac{\partial y_{0,t}}{\partial \beta_t} \left[ \delta_t \frac{ \frac{\partial B^{(t)}}{\partial
y} (x_{0,t},y_{0,t})}{B^{(t)}(x_{0,t},y_{0,t})} - \frac{\beta_t}{y_{0,t}} \right] - \log \left(
\frac{1 - \beta \int \lambda}{\beta \int \lambda} \right) \\
+ \frac{\partial z_0}{\partial \beta_t} \left[ \left( \frac{\int \rho}{\int \lambda} \right) \sum_{s
\in I_c} \gamma_s \frac{\frac{\mathrm{d} A^{(s)}}{\mathrm{d} z}(z_0)}{A^{(s)}(z_0)} -
\frac{\beta}{z_0} \right] = 0 \; .
\end{multline}
The terms in square brackets are zero due to (\ref{eq:x0t_y0t_eqn_1}), (\ref{eq:x0t_y0t_eqn_2}) and
(\ref{eq:z0_eqn}) respectively; therefore this simplifies to
\begin{equation}
z_0 y_{0,t} \left( \frac{1 - \beta \int \lambda}{\beta \int \lambda} \right) = 1 \quad \forall t \in
I_v \; .
\label{eq:z0_y0_beta_1}
\end{equation}
We conclude that all of the $\{ y_{0,t} \}$ are equal, and we may write 
\begin{equation}
y_{0,t} = y_0 \quad \forall t \in I_c \; .
\label{eq:y0_all_equal}
\end{equation} 
Rearranging (\ref{eq:z0_y0_beta_1}) we obtain (\ref{eq:z0_y0_relation}). Also, summing
(\ref{eq:x0t_y0t_eqn_1}) over $t \in I_v$ and using (\ref{eq:sum_alpha_constraint}) and
(\ref{eq:x0_all_equal}) yields (\ref{eq:x0_y0_eqn_1}). Similarly, summing (\ref{eq:x0t_y0t_eqn_2})
over $t \in I_v$ and using (\ref{eq:beta_sum_of_betat}) and (\ref{eq:y0_all_equal}) yields
(\ref{eq:x0_y0_eqn_2}). Substituting back into (\ref{eq:S_definition}) and using
(\ref{eq:x0_all_equal}), (\ref{eq:y0_all_equal}), (\ref{eq:sum_alpha_constraint}) and
(\ref{eq:beta_sum_of_betat}) yields 
\begin{multline}
G(\alpha) = \sum_{t \in I_v} \delta_t \log B^{(t)}(x_{0},y_{0}) - \alpha \log x_{0} - \beta \log
y_{0} \\
+ \left( \frac{\int \rho}{\int \lambda} \right) \sum_{s \in I_c} \gamma_s \log A^{(s)}(z_0) - \beta
\log z_0 - \frac{h(\beta \int \lambda)}{\int \lambda}
\label{eq:growth_rate_polynomial_general2}
\end{multline}
where $x_0$, $y_0$, $z_0$ and $\beta$ are the unique positive real solutions to the $4 \times 4$
system of equations (\ref{eq:z0_eqn}), (\ref{eq:x0_y0_eqn_1}), (\ref{eq:x0_y0_eqn_2}) and
(\ref{eq:z0_y0_relation}). Finally, (\ref{eq:z0_y0_relation}) leads to the observation that
\[
-\beta \log z_0 - \beta \log y_0 - \frac{h(\beta \int \lambda)}{\int \lambda} = \frac{\log \left( 1
- \beta \int \lambda \right)}{\int \lambda} 
\] 
which, when substituted in (\ref{eq:growth_rate_polynomial_general2}), leads to
(\ref{eq:growth_rate_polynomial_general}).


\section{Proofs of Lemmas in Section~\ref{section:spectral_shape_formula}}\label{section:supporting_proofs}

\indent\emph{Proof of Lemma~\ref{lemma:f_properties}}: We prove the second property, as the proofs of
the first and the third properties are straightforward. The derivative of $\sff$ (normalized w.r.t.
$\int\! \rho$) is given by
\begin{align*} 
\sum_{t \in I_c} \gamma_t \, \frac{A^{(t)}(z) \left[ \frac{\mathrm{d} A^{(t)}(z)}{\mathrm{d} z} + z
\, \frac{\mathrm{d}^2 A^{(t)}(z)}{\mathrm{d} z^2} \right] - z \left[ \frac{\mathrm{d}
A^{(t)}(z)}{\mathrm{d} z} \right] ^2}{[A^{(t)}(z)]^2} \, .
\end{align*}
The denominator of the fraction in each term in the sum is strictly positive for all $z > 0$. The
numerator of the fraction in term $t \in I_c$ in the sum may be expanded as
\begin{align*}
\, & (1+\sum_{v=r_t}^{s_t} A^{(t)}_v z^v) (\sum_{u=r_t}^{s_t} u A^{(t)}_u z^{u-1} +
\sum_{u=r_t}^{s_t} u(u-1) A^{(t)}_u z^{u-1})\\ \, & \phantom{-}\, - z (\sum_{u=r_t}^{s_t} u
A^{(t)}_u z^{u-1}) (\sum_{v=r_t}^{s_t} v A^{(t)}_v z^{v-1}) \\
\, & = \sum_{u=r_t}^{s_t} u^2 A^{(t)}_u z^{u-1} + \sum_{u=r_t}^{s_t}\sum_{v=r_t}^{s_t} u(u-v)
A^{(t)}_u A^{(t)}_v z^{u+v-1} \, .
\end{align*}
Observe that in this expression, each term in the second summation with $u=v$ is zero, while each
$(u,v)$ term in the second summation (with $u > v$) added to the corresponding $(v,u)$ term is
positive for $z>0$, since $u(u-v) A^{(t)}_u A^{(t)}_v z^{u+v-1} + v(v-u) A^{(t)}_u A^{(t)}_v
z^{u+v-1} = (u-v)^2 A^{(t)}_u A^{(t)}_v z^{u+v-1}> 0$ and therefore the second summation is
nonnegative for $z>0$. Since the first summation is strictly positive for $z>0$, it follows that
$\sff'(z)>0$ for all $z > 0$. \hfill \QEDclosed

\medskip
\indent\emph{Proof of Lemma~\ref{lemma:all_1_codeword}}: The sufficient condition (if a linear block
code has the all-$1$ codeword then its WEF is symmetric)
is a well-known result in classical coding theory. Next, we provide a proof for the necessary
condition. 

Reasoning by contradiction, assume that CN type $t \in I_c$ has a symmetric WEF $A^{(t)}(z)$ and a
maximum codeword weight $\bar{u}_t<s_t$. Note that, under these hypotheses, we necessarily have
$A^{(t)}_{\bar{u}_t}=1$. Without
loss of generality, we assume that the (unique) codeword of weight $\bar{u}_t$ has all its $1$
entries in the $\bar{u}_t$ leftmost positions. We build a generator matrix for the
component code associated with the CN in the form $\mbox{\boldmath $G$}^{(t)}=[\mbox{\boldmath
$G$}^{(t)}_{1}, \mbox{\boldmath $G$}^{(t)}_{2}]$, where $\mbox{\boldmath $G$}^{(t)}_{1}$ is an
$h_t \times \bar{u}_t$ matrix whose first row is the all-$1$ vector, and $\mbox{\boldmath
$G$}^{(t)}_{2}$ is
an $h_t \times (s_t-\bar{u}_t)$ matrix whose first row is the all-$0$ vector. Note that, due
to the above-stated hypotheses, $\mbox{\boldmath $G$}^{(t)}_{1}$ has a unique all-$1$ row.  We
denote
by $A^{(t,1)}(z)$ the WEF associated with the matrix $\mbox{\boldmath $G$}^{(t)}_1$. Since
$\mbox{\boldmath $G$}^{(t)}_{1}$ has a unique all-$1$ row, $A^{(t,1)}(z)$ is symmetric, i.e.,
$A^{(t,1)}_u = A^{(t,1)}_{\bar{u}_t-u}$ for all $u\in\{0,\dots,\bar{u}_t\}$.

It is now convenient to represent the WEF of the original code as a ``balls and bins'' system.
More specifically, there are $s_t+1$ bins, each uniquely associated with a weight $u$, and
$2^{h_t}$ balls, each uniquely associated with an information word \mbox{\boldmath $v$} (and then
with the codeword $\mbox{\boldmath $v$} \mbox{\boldmath $G$}^{(t)}$). Bins are labeled
$0,1,\dots,s_t$, such that each label is equal
to the corresponding codeword weight. They are arranged with increasing labels from left to
right. The number of codewords in bin $u\in\{0,1,\dots,s_t\}$ is the number of codewords of
weight~$u$.

We consider filling the bins according to the following procedure. At the beginning, all bins are empty. Then, in the \emph{initial phase} the $2^{h_t}$ balls are placed into the bins according to the WEF $A^{(t,1)}(z)$ corresponding to $\mbox{\boldmath $G$}^{(t)}_1$. Note that no
balls are placed in bins with label in $\{\bar{u}_t+1,\dots,s_t\}$ and that, due to the symmetry of
$A^{(t,1)}(z)$, the number of balls in bins $u$ and $\bar{u}_t-u$ is the same for all $u\in\{0,\dots,\bar{u}_t\}$. Then, in the \emph{adjustment phase} the 
balls are moved in order to match the WEF $A^{(t)}(z)$, according to $\mbox{\boldmath$G$}^{(t)}_{2}$. The key observation here is that, in the adjustment phase, every ball must either stay in its current bin, or else move to the right. This
is because, for any information word {\boldmath $v$}, of length $h_t$, the Hamming weight of
$\mbox{\boldmath $v$} \mbox{\boldmath $G$}^{(t)}_{1}$ cannot be larger than the Hamming weight of
$\mbox{\boldmath $v$} \mbox{\boldmath $G$}^{(t)}$. 

Denoting by $u_{\min}$ the minimum weight such that $A^{(t,1)}_{u} \neq A^{(t)}_{u}$, we must have
$A^{(t,1)}_{u_{\min}} > A^{(t)}_{u_{\min}}$ (since no ball can move into bin $u_{\min}$, while at least one ball has moved out). 
Note that the number of balls in all bins corresponding to weights smaller than $u_{\min}$ remains unchanged. Since $A^{(t,1)}(z)$ is symmetric, and since
so is $A^{(t)}(z)$ by hypothesis, we must also have $A^{(t,1)}_{\bar{u}_t-u_{\min}} >
A^{(t)}_{\bar{u}_t-u_{\min}}$. 
It follows that the \emph{total} number of balls in bins with labels in $\{ \bar{u}_t-u_{\min}, \ldots, \bar{u}_t \}$ has \emph{decreased} during the adjustment phase. But this is a contradiction, and so the result is established.
\hfill \QEDclosed

\medskip
\indent\emph{Proof of Lemma~\ref{lemma:A_symmetry}}: Assume $A^{(t)}(z)$ is symmetric (and therefore $\bar{u}_t=s_t$). We have
\begin{align*}
z^{s_t} A^{(t)}\left(z^{-1}\right) &= \sum_{u = 0}^{s_t} A^{(t)}_u z^{s_t-u} \\
&= \sum_{v = 0}^{s_t} A^{(t)}_{s_t-v} z^{v} \\
&= A^{(t)}(z)
\end{align*}
where the final equality is due to $A^{(t)}_{s_t-v}=A^{(t)}_v$
for all $v \in \{ 0,1,\ldots,s_t \}$. Conversely, assume that \eqref{eq:symmetry_of_CN_type} is satisfied for all
$z \in \mathbb{R}^+$. It can be easily recast as %
$$\sum_{u = 0}^{s_t} A_u^{(t)} z^{u} = \sum_{u = 0}^{\bar{u}_t} A_{\bar{u}_t-u}^{(t)}
z^{u} \, .$$
In order for this equality to be satisfied for all $z \in \mathbb{R}^+$, we must have $A^{(t)}_{\bar{u}_t-u}=A^{(t)}_u$ for all $u \in \{ 0,1,\ldots,s_t \}$. Hence, $A^{(t)}(z)$ must be symmetric.
\hfill \QEDclosed

\medskip
\indent\emph{Proof of Lemma~\ref{lemma:finv_symmetry}}: %
We prove that the function $\sff$ fulfills
\begin{align}\label{eq:f_symmetry}
\sff(z) = \mathsf{M} - \sff\left(z^{-1}\right)
\end{align}
$\forall$ $z \in \mathbb{R}^+$ if $A^{(t)}(z)$ is symmetric for every $t \in I_c$. The
result is then obtained by applying the inverse function $\mathsf{f}^{-1}$ to both sides of
\eqref{eq:f_symmetry} and by letting $\sff(z^{-1})=\alpha$ for all $z \in \mathbb{R}^+ \backslash
\{0\}$.

Assuming that the WEF of CN type $t \in I_c$ is symmetric, we have,
differentiating (\ref{eq:symmetry_of_CN_type}),
\[
\frac{\mathrm{d} A^{(t)}(z)}{\mathrm{d} z} = - z^{\bar{u}_t-2} \frac{\mathrm{d}
A^{(t)}(z^{-1})}{\mathrm{d} z^{-1}} + \bar{u}_t z^{\bar{u}_t-1} A^{(t)}(z^{-1}) \, .
\]
Multiplying by $z$ and using (\ref{eq:symmetry_of_CN_type}) yields
\begin{equation}\label{eq:symmetry_of_CN_type_derivative}
z \frac{\mathrm{d} A^{(t)}(z)}{\mathrm{d} z} = - z^{\bar{u}_t-1} \frac{\mathrm{d}
A^{(t)}(z^{-1})}{\mathrm{d} z^{-1}} + \bar{u}_t A^{(t)}(z) \; .
\end{equation}
Then,
\begin{align*}
\mathsf{M} - \sff\left(z^{-1} \right) & = \mathsf{M} - \left( \smallint\! \rho \right) \sum_{t \in
I_c} \gamma_t \left( \frac{z^{-1} \frac{\mathrm{d} A^{(t)}(z^{-1})}{\mathrm{d}
z^{-1}}}{A^{(t)}(z^{-1})} \right) \\
\, & \!\!\!\!\!\!\!\!\!\!\!\!\!\!\!\!\!\!\!\!\!\!\!\!\!\!\!\!\!\!\!\!\stackrel{\textrm{(a)}}{=}
\left( \smallint\! \rho \right) \sum_{t \in I_c} \gamma_t \left( \bu_t - \frac{z^{\bar{u}_t-1}
\frac{\mathrm{d} A^{(t)}(z^{-1})}{\mathrm{d} z^{-1}}}{A^{(t)}(z)} \right) \\
\, & \!\!\!\!\!\!\!\!\!\!\!\!\!\!\!\!\!\!\!\!\!\!\!\!\!\!\!\!\!\!\!\!\stackrel{\textrm{(b)}}{=}
\left( \smallint\! \rho \right) \sum_{t \in I_c} \gamma_t \frac{\bu_t A^{(t)}(z) + z
\frac{\mathrm{d} A^{(t)}(z)}{\mathrm{d} z} - \bar{u}_t A^{(t)}(z)}{A^{(t)}(z)}=\sff(z)
\end{align*}
where we have used \eqref{eq:M_definition} and \eqref{eq:symmetry_of_CN_type} in (a), and
\eqref{eq:symmetry_of_CN_type_derivative} in (b). 
\hfill \QEDclosed


\section{Examples}
\label{section:examples}
In this section, the spectral shapes of some example GLDPC and D-GLDPC ensembles are evaluated using
the polynomial solution of Theorem \ref{thm:growth_rate}. We will consider both BD and MAP CN decoding. Considering the former case, note that under bounded distance decoding, the non-empty stopping sets are precisely those sets which have $r_t$ or more erased code bits; therefore the local SSEF (BD-SSEF) is given by
\begin{align}\label{eq:Psi_def}
\Psi^{(t)}(z) = 1 + \sum_{u=r_t}^{s_t} {s_t \choose u} z^u \; .
\end{align}
In the latter case, the local SSEF (MAP-SSEF) is given by
\begin{align}\label{eq:Phi_def}
\Phi^{(t)}(z) = 1 + \sum_{u=r_t}^{s_t} \phi_u^{(t)}\, z^u
\end{align}
where $\phi^{(t)}_u \geq 0$ is the number of local stopping sets (under MAP decoding) of size
$u$.\footnote{Denoting by $\mathbf{G}_t$ any generator matrix for a type-$t$ CN, a local erasure
pattern is a local stopping set under MAP decoding when each column of $\mathbf{G}_t$ corresponding
to erased bits is linearly independent of the columns of $\mathbf{G}_t$ corresponding to the
non-erased bits.} Furthermore, numerical examples are presented on the approximation of the parameter 
$\alpha^*$ for regular LDPC code ensembles, based on \eqref{eq:alpha_star_approx_LDPC_regular}.
\begin{table}[!t]
\caption{Coefficients of $\lambda(x)$ and $\rho(x)$ for the two example D-GLDPC ensembles.}\label{table:ensembles}
\begin{center}
\begin{tabular}{llll}
\hline\hline
\multicolumn{4}{c}{\emph{Ensemble $1$}}\\
\hline
\multicolumn{2}{l}{\emph{Variable nodes}} & \multicolumn{2}{l}{\emph{Check nodes}}\\
1:repetition$-2$ & $\lambda_1=0.055646$ & 1:Hamming$(7,4)$ & $\rho_1=0.965221$\\
2:SPC$-7$ (C) & $\lambda_2=0.944354$ & 2:SPC$-7$ & $\rho_2=0.034779$\\
\hline
\multicolumn{4}{c}{\emph{Ensemble $2$}}\\
\hline
\multicolumn{2}{l}{\emph{Variable nodes}} & \multicolumn{2}{l}{\emph{Check nodes}}\\
1:repetition$-2$ & $\lambda_1=0.022647$ & 1:Hamming$(7,4)$ & $\rho_1=0.965221$\\
2:SPC$-7$ (C) & $\lambda_2=0.100000$ & 2:SPC$-7$ & $\rho_2=0.034779$\\
3:SPC$-7$ (A) & $\lambda_2=0.539920$ &  & \\
4:SPC$-7$ (S) & $\lambda_2=0.337432$ &  & \\
\hline \hline
\end{tabular}
\end{center}
\end{table}
\begin{figure}
\begin{center}
\psfrag{xlabel}[c]{\small{$\phantom{---}\alpha$}} \psfrag{ylabel}[c]{\small{$\phantom{---}$\textsf{Growth rate, }$G(\alpha)$}} \psfrag{Ensemble 1 (2 VN types)}[l]{\begin{picture}(0,0)\put(0.2,0.2){\scriptsize{\textsf{Ensemble $1$ ($2$ VN types)}}}\end{picture}} \psfrag{Ensemble 2 (4 VN types)}[l]{\begin{picture}(0,0)\put(0.2,0.2){\scriptsize{\textsf{Ensemble $2$ ($4$ VN types)}}}\end{picture}}
\includegraphics[%
  width=1.0\columnwidth,
  keepaspectratio]{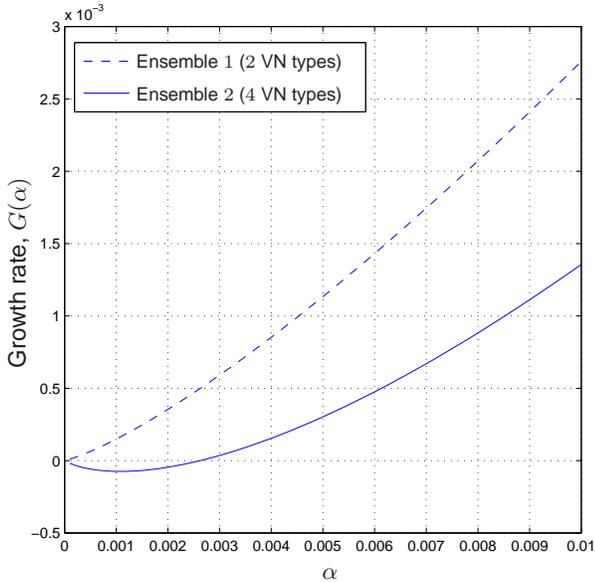}
\end{center}
\caption{\label{cap:growth_rates_paper} Growth rates of the two example ensembles described in Example~\ref{example:0}. Ensemble $1$ has bad growth rate behavior, while Ensemble $2$ has good growth rate behavior with a critical exponent codeword weight ratio of $\alpha^* = 2.625 \times 10^{-3}$.}
\end{figure}
\begin{example}[D-GLDPC ensembles with Hamming CNs and SPC VNs]\label{example:0} In this first example, we design two ensembles with design rate $R=1/2$ using Hamming $(7,4)$ codes as generalized CNs and SPC codes as generalized VNs. Three representations of SPC VNs are considered, namely, the cyclic (C), the systematic (S) and the antisystematic (A) representations \footnote{The $(k \times (k+1))$ generator matrix of a SPC code in A form is obtained from the generator matrix in S form by complementing each bit in the first $k$ columns. Note that a $(k \times (k+1))$ generator matrix in A form represents a SPC code if and only if the code length $q=k+1$ is odd. For even $k+1$ we obtain a $d_{\min}=1$ code with one codeword of weight $1$.}. 

Ensemble $1$ is characterized by two CN types and two VN types. Specifically, we have $I_c = \{ 1,2 \}$, where $1 \in I_c$ denotes a $(7,4)$ Hamming CN type and $2 \in I_c$ denotes a length-$7$ single parity check (SPC) CN type, and $I_v = \{ 1,2 \}$, where $1 \in I_v$ denotes a repetition-$2$ VN type and $2 \in I_v$ denotes a length-$7$ SPC CN type in cyclic form. 
Ensemble $2$ is characterized by two CN types and four VN types. Specifically, we have $I_c = \{ 1,2 \}$, where $1 \in I_c$ denotes a $(7,4)$ Hamming CN type and $2 \in I_c$ denotes a SPC-$7$ CN type, and $I_v = \{ 1,2,3,4 \}$, where $1 \in I_v$ denotes a repetition-$2$ VN type, $2 \in I_v$ denotes a length-$7$ SPC CN type in cyclic form, $3 \in I_v$ denotes a length-$7$ SPC CN type in antisystematic form, and $4 \in I_v$ denotes a length-$7$ SPC CN type in systematic form. 
The edge-perspective type distributions of the two ensembles are summarized in Table~\ref{table:ensembles}.

Both Ensemble 1 and Ensemble 2 have been obtained by performing a decoding threshold optimization with differential evolution \cite{storn05:differential-book}. This is an evolutionary parallel optimization algorithm to find the global minimum of a real-valued cost function of a vector of continuous parameters, where the cost function may even be defined by a procedure (e.g., density evolution returning the threshold for an LDPC code ensemble over some channel and under some decoding algorithm). It is based on the evolution of a population of $N_p$ vectors, and its main steps are the same as typical evolutionary optimization algorithms (\emph{mutation}, \emph{crossover}, \emph{selection}) \cite{citeulike:3401906}. A starting population of $N_p$ vectors is first generated. Then, a competitor (or \emph{trial vector}) for each population element is generated by combining a subset of randomly chosen vectors from the same population. Finally, each element of the population is compared with its trial vector and the vector yielding the smallest cost function value is selected as the corresponding element of the evolved population.\footnote{Usually, the algorithm's ``greediness'' is reduced by selecting the vector yielding the smallest cost with a probability smaller than one, instead of systematically selecting it.} These steps are iterated until a certain stopping criterion is fulfilled or until a maximum number of iterations is reached. Differential evolution was first proposed for the optimization of LDPC code degree profiles in \cite{shokrollahi00:design}.

In our experiments, each element of the population was a pair of polynomials $(\lambda,\rho)$ corresponding to given variable and check component codes, given VN representations, and given design rate $R$, while the cost function was the ensemble threshold over the BEC (returned by numerical procedure), under iterative decoding with MAP erasure decoding at the VNs and CNs.

For Ensemble $1$ we have $C \cdot V = 1.19 > 1$, so the ensemble has \emph{bad growth rate behavior} ($\alpha^*=0$). Ensemble $2$ has been obtained by imposing the further constraints $C \cdot V \leq 0.5$ and $\lambda_2 \geq 0.1$ on differential evolution optimization. Since in this case we have $C \cdot V = 0.5 < 1$, the ensemble has \emph{good growth rate behavior} ($\alpha^*>0$). The expected good or bad growth rate behavior of the two ensembles is reflected in the growth rate curves shown in Fig.~\ref{cap:growth_rates_paper}. Using a standard numerical solver, it took only $5.1$ s and $6.7$ s to evaluate $100$ points on the Ensemble $1$ curve and on the Ensemble $2$ curve, respectively. The relative minimum distance of Ensemble $2$ is \mbox{$\alpha^* = 2.625 \times 10^{-3}$}.
\end{example}
\begin{figure}
\begin{center}
\psfragscanon
\psfrag{xaxis}[lt]{\scriptsize{$\alpha$}}
\psfrag{yaxis}[b]{\scriptsize{\textsf{Growth rate,} $G(\alpha)$}}
\includegraphics[%
  width=0.8\columnwidth,
  keepaspectratio]{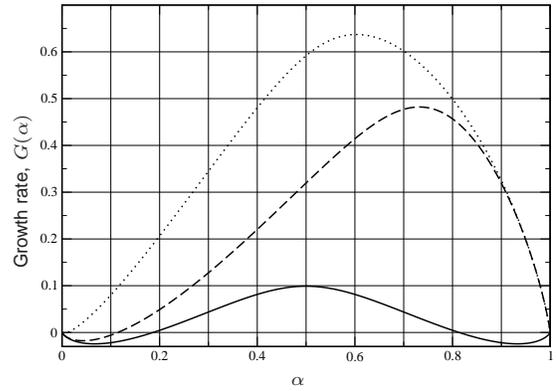}
\end{center}
\caption{Spectral shapes of the Tanner code ensemble in Example~\ref{example:1}. Solid: weight
spectral shape (critical exponent codeword weight ratio: $\alpha^* = 0.18650$). Dashed: stopping set
size spectral shape under MAP decoding at the CNs (relative minimum stopping set size:
$\alpha_{\Phi}^* = 0.11414$). Dotted: stopping set size spectral shape under BD decoding at the CNs
(critical exponent stopping set size ratio: $\alpha_{\Psi}^* = 0.01025$).}
\label{fig:growth_rate_rep2_H7}
\end{figure}
\begin{figure}
\begin{center}
\psfragscanon
\psfrag{xaxis}[lt]{\scriptsize{$\alpha$}}
\psfrag{yaxis}[b]{\scriptsize{\textsf{Growth rate,} $G(\alpha)$}}
\includegraphics[%
  width=0.8\columnwidth,
  keepaspectratio]{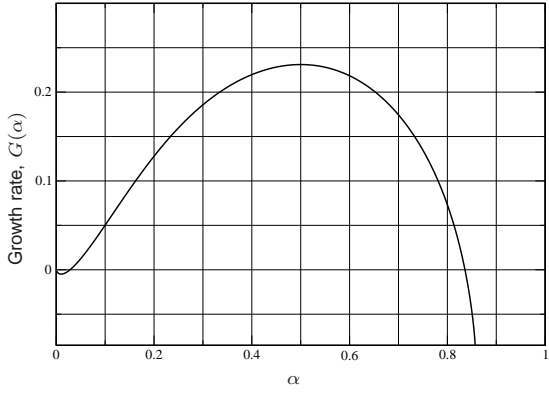}
\end{center}
\caption{Weight spectral shape of the check-hybrid GLDPC code ensemble in Example~\ref{example:2}. Critical exponent codeword weight ratio: $\alpha^* = 0.028179$.}
\label{fig:growth_rate_check_hybrid}
\end{figure}

\begin{figure}[h]
\begin{center}
\vspace{4mm}
\psfragscanon
\psfrag{xaxis}[lt]{\scriptsize{$\alpha$}}
\psfrag{yaxis}[b]{\scriptsize{\textsf{Growth rate,} $G(\alpha)$}}
\includegraphics[%
  width=0.8\columnwidth,
  keepaspectratio]{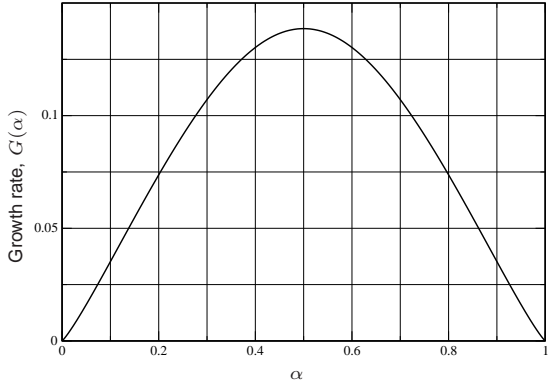}
\end{center}
\caption{Weight spectral shape of the Tanner code ensemble in Example~\ref{example:3}.}
\label{fig:growth_rate_asymp_bad}
\end{figure}

\begin{example}[Tanner code with $(7,4)$ Hamming CNs]\label{example:1} Consider a rate $R=1/7$
Tanner code ensemble where all VNs have degree $2$ and where all CNs are $(7,4)$ Hamming codes
(it was shown in \cite{lentmaier99:generalized,boutros99:generalized} that this ensemble has good
growth rate
behavior). The WEF of a Hamming $(7,4)$ CN is given by $A(z)=1 + 7 z^3 + 7 z^4 + z^7$, while its
local MAP-SSEF and BD-SSEF are given by $\Phi(z)=1+7 z^3+10 z^4+21 z^5+7 z^6+z^7$ and $\Psi(z)=1+35
z^3+35 z^4+21 z^5+7 z^6+z^7$ respectively. Note that we have $\mathsf{M}=\frac{\bu}{s}=1$ in all
three cases.
A plot of $G(\alpha)$, $G_{\Phi}(\alpha)$ and $G_{\Psi}(\alpha)$ obtained by implementation of
\eqref{eq:G(alpha)_tanner_codes} is depicted in Fig.~\ref{fig:growth_rate_rep2_H7}. We observe that
$A(z)$ satisfies the conditions of Theorem~\ref{theorem:G_symmetry_sufficient}. This is reflected by
the fact that the weight spectral shape $G(\alpha)$ is symmetric with respect to $\alpha=1/2$.
\end{example}
\begin{example}[Check-hybrid ensemble]\label{example:2} Consider a rate $R=1/3$ check-hybrid GLDPC
code ensemble where all VNs are repetition codes of length $q=3$ and whose CN set is composed of a
mixture of two linear block code types ($I_c=\{1,2\}$). CNs of type $1 \in I_c$ are length-$7$ SPC
codes with WEF $A^{(1)}(z)=[(1+z)^7+(1-z)^7]/2$ and $\gamma_1=0.722$, while CNs of type $2 \in I_c$
are $(7,4)$ codes with WEF $A^{(2)}(z)=1+5 z^2+7 z^4+3 z^6$ and $\gamma_2=0.278$. The weight
spectral shape of this ensemble, obtained from \eqref{eq:G(alpha)_irregular_CN_set}, is depicted in
Fig.~\ref{fig:growth_rate_check_hybrid}. Note that for this ensemble $\mathsf{M}=6/7$, and also that
the weight spectral shape does not exhibit any symmetry property (the CN WEFs are not symmetric).
\end{example}

\begin{example}[Ensemble with bad growth rate behavior]\label{example:3}
Consider a rate $R=1/5$ Tanner code ensemble where all VNs are repetition codes of length $q=2$ and
and where all CNs are $(5,3)$ linear block codes with WEF $A(z)=1+3 z^2+3 z^3+z^5$. This ensemble is
known to have bad growth rate behavior ($\alpha^*=0$) since we have $\lambda'(0)C=6/5>1$, where
$\lambda(x)=x$ and $C=2 A_2 / s$ \cite{Tillich04:weight,paolini08:weight}. A plot of the weight
spectrum for this ensemble, obtained from \eqref{eq:G(alpha)_tanner_codes} is depicted in
Fig.~\ref{fig:growth_rate_asymp_bad}. We observe that the plot of $G(\alpha)$ is symmetric, due to
the fact that $A(z)$ is symmetric ($\mathsf{M}=1$). As expected, the derivative of $G(\alpha)$ at
$\alpha=0$ is positive and hence $\alpha^*=0$.
\end{example}

\begin{figure}[t]
\psfrag{alpha_star}[b]{\small{\quad$\alpha^*$}} \psfrag{dv}[]{\small{$d_v$}}
\begin{center}\includegraphics[%
  width=\columnwidth,
  keepaspectratio]{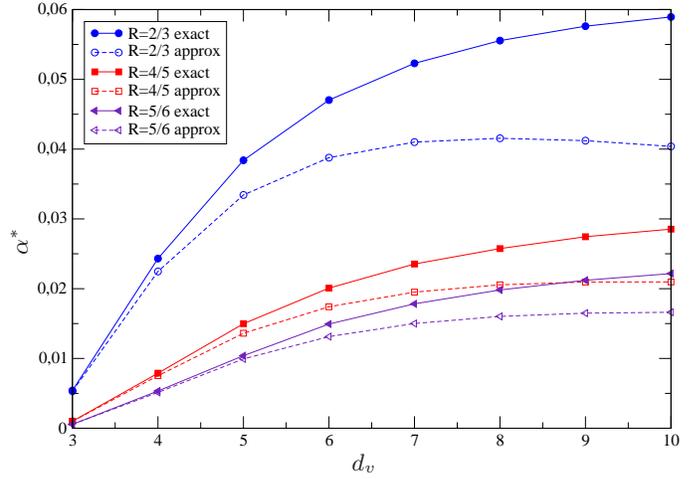}
  \end{center}
\caption{Actual values of $\alpha^*$, and the corresponding approximated values based on \eqref{eq:alpha_star_approx_LDPC_regular}, for some regular
LDPC code ensembles.}\label{fig:alpha_star}
\end{figure}

Finally, note that for the ensembles of Figures \ref{fig:growth_rate_rep2_H7}--\ref{fig:growth_rate_asymp_bad}, the weight spectral shape has in each case a maximum value of $R \log 2$ which occurs at $\alpha=1/2$, in accordance with Theorem~\ref{thm:max_of_weight_spectral_shape}.
\begin{example}
Fig.~\ref{fig:alpha_star} illustrates a comparison between the actual value of the parameter $\alpha^*$
for some high-rate regular LDPC code ensembles (solid curves) and the corresponding approximation
\eqref{eq:alpha_star_approx_LDPC_regular} obtained via growth rate analysis in the small-$\alpha$
case (dashed curves). Each curve corresponds to a value of the design rate
$R$, so the CN degree $d_c$ for each point is given by $d_c=d_v/(1-R)$ where the VN
degree $d_v$ is reported in abscissa. From the figure we see that the approximation is quite good
for small values of $d_v$, which are usually the ones of interest in practical applications. Moreover, the higher the design rate $R$ the larger the range of $d_v$ over which
the approximation is satisfactory. The actual values of $\alpha^*$ and their corresponding approximated values are reported in Table \ref{table:alpha} for several regular LDPC code ensembles with $d_v=3$. A very good match can be observed.  Similar examples may be developed, for instance, for Tanner codes through \eqref{eq:alpha_star_approx_GLDPC_variable_regular}.
\end{example}

\begin{table}[!t]
\caption{Comparison between the actual values of the parameter $\alpha^*$ for some regular LDPC ensembles with VN degree $d_v=3$ and the corresponding approximated values based on \eqref{eq:alpha_star_approx_LDPC_regular}.}\label{table:alpha}
\begin{center}
\begin{tabular}{lll}
\hline\hline
$R$, $d_c$ & $\alpha^*$, exact & $\alpha^*$, approx. \\
\hline
$1/4$, $4$ & $0.112159$ & $0.100677$ \\
$2/5$, $5$ & $0.045365$ & $0.042473$ \\
$1/2$, $6$ & $0.022733$ & $0.021746$ \\
$4/7$, $7$ & $0.012993$ & $0.012585$ \\
$5/8$, $8$ & $0.008117$ & $0.007925$ \\
$2/3$, $9$ & $0.005410$ & $0.005309$ \\
$7/10$, $10$ & $0.003785$ & $0.003729$ \\
\hline \hline
\end{tabular}
\end{center}
\end{table}

\section{Conclusion}\label{section:conclusion}
A general expression for the weight and stopping set size spectral shapes of irregular D-GLDPC
ensembles has been presented. Evaluation of the expression requires solution of a $4 \times 4$
polynomial system, irrespective of the number of VN and CN types in the ensemble. A compact
expression was developed for the special case of check-hybrid GLDPC codes, and both a necessary and
a sufficient condition for symmetry of the weight spectral shape was developed. Simulation results
were presented for two example optimized irregular D-GLDPC code ensembles as well as a number of
check-hybrid GLDPC code ensembles.
\appendices

\section{Some Useful Lemmas}
The following results are special cases of \cite[Corollary 16]{di06:weight}.
\medskip
\begin{lemma}
Let $A(x) = 1 + \sum_{u=c}^{d} A_u x^u$, where $1 \le c \le d$, be a polynomial satisfying $A_c > 0$
and $A_u \ge 0$ for all $c < u \le d$. Then, as $\ell \rightarrow \infty$,
\begin{equation}
\coeff \left[ \left( A(x) \right) ^{\ell}, x^{\xi \ell} \right] \asympequal \exp \left[ \ell \log
\left( \frac{A(z)}{z^{\xi}} \right) \right]
\end{equation}
where $z$ is the unique positive real solution to
\begin{equation}
\label{eq:z_soln_to_A_eqn}
\frac{A'(z)}{A(z)} \cdot z = \xi \; .
\end{equation}
\label{lemma:optimization_1D}
\end{lemma} 
\medskip
\begin{lemma}
Let 
\[
B(x,y) = 1 + \sum_{u=1}^{k} \sum_{v=c}^{d} B_{u,v} x^u y^v 
\]
where $k \ge 1$ and $1 \le c \le d$, be a bivariate polynomial satisfying $B_{u,v} \ge 0$ for all $1
\le u \le k$, $c \le v \le d$. Then, as $\ell \rightarrow \infty$,
\begin{equation}
\coeff \left[ \left( B(x,y) \right) ^{\ell}, x^{\xi \ell} y^{\theta \ell} \right] \asympequal \exp
\left[ \ell \log \left( \frac{B(x_0,y_0)}{x_0^{\xi} y_0^{\theta}} \right) \right]
\end{equation}
where $x_{0}$ and $y_{0}$ are the unique positive real solutions to the pair of simultaneous equations
\begin{equation}
\frac{ \frac{\partial B}{\partial x} (x_{0},y_{0})}{B(x_{0},y_{0})} \cdot x_{0} = \xi
\end{equation}
and
\begin{equation}
\frac{ \frac{\partial B}{\partial y} (x_{0},y_{0})}{B(x_{0},y_{0})} \cdot y_{0} = \theta \; . 
\end{equation}
\label{lemma:optimization_2D}
\end{lemma}


\section{Solution for Small Linear-Weight Codewords}
\label{app:small_alpha}
In this appendix, we analyze Theorem \ref{thm:growth_rate} for the case of small $\alpha$.
Specifically, we prove Corollary \ref{corollary:small_alpha}, which is a slightly weaker form of \cite[Theorem 4.1]{flanagan09:IEEE-IT}.
The proof consists of first obtaining expressions for $z_0$, $y_0$ and $x_0$ in terms of
$\beta$, and for $\beta$ in terms of $\alpha$, and then exploiting these in
\eqref{eq:growth_rate_polynomial_general}.

First we develop an expression for $z_0$ in terms of $\beta$. Considering \eqref{eq:x0_y0_eqn_1},
we have that its left-hand side must be $o(1)$ because so is its right-hand side. The only possibility is that $x_0^i
y_0^j=o(1)$ for some admissible choices of $(i,j,t)$, and $x_0^i y_0^j=o(\alpha)$ for all other
admissible choices. Since the only difference between the left-hand side of \eqref{eq:x0_y0_eqn_1} and the left-hand side
of \eqref{eq:x0_y0_eqn_2} is represented by the coefficients of the $x_0^i y_0^j$ terms, it follows
that
\begin{align}\label{eq:beta=o_alpha}
\lim_{\alpha\rightarrow0}\beta = 0\, .
\end{align}
For the moment, the notation $o(1)$ will be intended as $\beta\rightarrow 0$. (We will show later
that $\beta$ is proportional to $\alpha$ to the first order, so that the notation $o(\beta)$ is
equivalent to the notation $o(\alpha)$.) Because of the above discussion, the left-hand side of \eqref{eq:z0_eqn}
must also be $o(1)$, a condition which can
be satified if and only if $z_0^u=o(1)$ for some admissible choices of $(u,t)$ and $z_0^u=o(\beta)$
for all other admissible choices of $(u,t)$. Consider now \eqref{eq:z0_eqn} written as an equality
of polynomials. Due to \eqref{eq:beta=o_alpha}, its left-hand side is dominated by the terms corresponding to
$u=r$, and the equation can be written in the form
\begin{align*}
z_0^r  \sum_{t \; : \; r_t = r} \frac{\rho_t r A_r^{(t)}}{s_t} = \beta \smallint\!\lambda \, (1 +
o(1)) \; ,
\end{align*}
i.e.,\footnote{Here (and often throughout the proof) we use the property
$(1+o(1))^{k}=1+o(1)$ for rational $k$.}
\begin{align}
\label{eq_z0_beta}
z_0 = \left( \frac{\beta \int\!\lambda}{C} \right)^{1/r} (1+o(1))\, .
\end{align}

We next develop an expression for $y_0$ in terms of $\beta$; from
\eqref{eq:z0_y0_relation} we have 
$$
y_0=\frac{1}{z_0} \cdot \frac{\beta\int\!\lambda}{1-\beta\int\!\lambda}
$$
which, combined with \eqref{eq_z0_beta} and \eqref{eq:beta=o_alpha} respectively, yields
\begin{align}
\label{eq:y0_beta}
y_0 = C^{1/r}\left(\beta\smallint\!\lambda\right)^{1/\psi}(1+o(1)) \, .
\end{align}

Next, we develop an expression for $x_0$ in terms of $\beta$. To this purpose, since $x_0^i
y_0^j=o(1)$ for some admissible choices of $(i,j)$, \eqref{eq:x0_y0_eqn_2} (when written as
an equality of polynomials) can be expressed as
$$
\sum_{t\in I_v} \frac{\lambda_t}{q_t \int\!\lambda} \sum_{(i,j) \in S_t^{-}}
j\,B_{i,j}^{(t)}\, x_0^i y_0^j = \beta (1+o(1))
$$
where $S_t^{-} = \{ (i,j) \neq (0,0) \: : \: B_{i,j}^{(t)} > 0 \}$. Using
\eqref{eq:y0_beta}, this latter equation can be written in the form
\begin{align}
\label{eq:pre_x0_beta}
\sum_{t\in I_v} \frac{\lambda_t}{q_t} \sum_{(i,j) \in S_t^{-}} j\,B_{i,j}^{(t)}\,
C^{\frac{j}{r}} \left(x_0 \left( \beta\smallint\!\lambda\right)^{\frac{T_{i,j}}{\psi}}\right)^i = 1 + o(1)\, ,
\end{align}
where $T_{i,j}=(j-\psi)/i \ge 0$. It follows that, as $\beta\rightarrow 0$,
the left-hand side of \eqref{eq:pre_x0_beta} is dominated by summands corresponding to admissible choices of
$(i,j)$ for which $x_0 \beta^{T_{i,j}/\psi}$ tends to a constant.
Letting $T = \min\{T_{i,j} \; : \; t \in I_v \textrm{ and } (i,j) \in S_t^{-} \}$, these
dominating summands necessarily correspond to admissible choices of
$(i,j)$ such that $T_{i,j}=T$. In fact, assume that $x_0 \beta^{T_{i,j}/\psi}=c+o(1)$ for some
$T_{i,j}>T$ and constant $c$. Then, for all choices of $(i,j)$ such that $T_{i,j}=T$, the term
$x_0 \beta^{T_{i,j}/\psi}$ would be unbounded, contradicting \eqref{eq:pre_x0_beta}. Hence,
we have
$$
\sum_{t\in Y_v} \frac{\lambda_t}{q_t} \sum_{(i,j)\in P_t} j\,B_{i,j}^{(t)}\, C^{j/r}
\left(x_0 \left( \beta\smallint\!\lambda\right)^{\frac{T}{\psi}}\right)^i = 1+o(1)\, ,
$$
i.e.,
$$
Q_1 \left( x_0 (e \beta)^{T/\psi} \right) = 1 + o(1)
$$
and therefore\footnote{Note that, using the Taylor series of $Q_1^{-1}(1+x)$ around $x=0$, we have
$Q_1^{-1}(1+o(1))=Q_1^{-1}(1)+o(1)$.}
\begin{align}\label{eq:x0_beta}
x_0 = \frac{Q_1^{-1}(1)}{e^{T/\psi}}\beta^{-T/\psi} (1+o(1)) \; .
\end{align}
Next, we develop an expression for $\beta$ in terms of $\alpha$. 
Similarly to \eqref{eq:x0_y0_eqn_2}, \eqref{eq:x0_y0_eqn_1} (when written as an equality of
polynomials) can be expressed as
$$
\sum_{t\in I_v} \frac{\lambda_t}{q_t\int\!\lambda} \sum_{(i,j) \in S_t^{-}}
i\,B_{i,j}^{(t)}\, x_0^i y_0^j = \alpha (1+o(1))
$$
where now $o(1)$ is intended as $\alpha\rightarrow 0$. 
Using \eqref{eq:y0_beta} and dividing each side by $\beta$ we obtain
$$
\sum_{t\in I_v} \frac{\lambda_t}{q_t} \sum_{(i,j) \in S_t^{-}} i\,B_{i,j}^{(t)}\,
C^{j/r} \left(x_0 \left( \beta\smallint\!\lambda\right)^{\frac{T_{i,j}}{\psi}}\right)^i =
\frac{\alpha}{\beta} (1+o(1))\, .
$$
Reasoning in the same way as we did for \eqref{eq:pre_x0_beta}, we can write the previous equation
in the form
$$
Q_2 \left( x_0 (e \beta)^{T/\psi} \right) = \frac{\alpha}{\beta}(1+o(1))
$$
which, using \eqref{eq:x0_beta}, becomes\footnote{Note that
$Q_2(Q_1^{-1}(1)(1+o(1)))=Q_2(Q_1^{-1}(1))+o(1)$ using the Taylor series of $Q_2(Q_1^{-1}(1)+x)$ around $x=0$.}
$$
Q_2(Q_1^{-1}(1))=\frac{\alpha}{\beta}(1+o(1))\, .
$$
This yields
\begin{align}\label{eq:beta_alpha}
\beta=\frac{\alpha}{Q_2(Q_1^{-1}(1))}(1+o(1))\, .
\end{align}
We conclude from \eqref{eq:beta_alpha} that $\beta$ is proportional to $\alpha$ to the first
order, so that $o(\beta)$ and $o(\alpha)$ are equivalent notations.

We next use the derived expressions to analyze \eqref{eq:growth_rate_polynomial_general}.
Using the Taylor series of $\log(1+x)$ around $x=0$, we have
\begin{align}\label{eq:G_fourth_term}
\frac{\log(1-\beta\int\!\lambda)}{\int\!\lambda}=-\beta(1+o(1))\, .
\end{align}
Consider now the term $\frac{\int\!\rho}{\int\!\lambda}\sum_{t\in I_c}\gamma_t \log(A^{(t)}(z_0))$.
Because of \eqref{eq_z0_beta}, we have
$$
A^{(t)}(z_0)= \left\{ \begin{array}{lll}
1+A_r^{(t)}\frac{\beta\int\!\lambda}{C}(1+o(1)) & \quad \textrm{if} & r_t = r \\
1+o(\beta) & \quad \textrm{if} & r_t > r
\end{array} \right.
$$
and, using the Taylor series of $\log(1+x)$ around $x=0$,
$$
\log A^{(t)}(z_0)= \left\{ \begin{array}{lll}
A_r^{(t)}\frac{\beta\int\!\lambda}{C}(1+o(1)) & \quad \textrm{if} & r_t = r \\
o(\beta) & \quad \textrm{if} & r_t > r \, .
\end{array} \right.
$$
This yields
\begin{align*}
\frac{\int\!\rho}{\int\!\lambda}\sum_{t\in I_c}\gamma_t \log(A^{(t)}(z_0)) & =
\sum_{t \; : \; r_t = r}  \frac{\rho_t}{s_t}
A_r^{(t)}\frac{\beta}{C}(1+o(1))
\end{align*}
and finally, recalling \eqref{eq:general_C_definition1},
\begin{align}\label{eq:G_third_term}
\frac{\int\!\rho}{\int\!\lambda}\sum_{t\in I_c}\gamma_t \log(A^{(t)}(z_0)) =
\frac{\beta}{r}(1+o(1))\, .
\end{align}
We then have
\begin{align}\label{eq:G_fourth_plus_third}
\frac{\int\!\rho}{\int\!\lambda}\sum_{t\in
I_c}\gamma_t \log(A^{(t)}(z_0)) + \frac{\log(1-\beta\int\!\lambda)}{\int\!\lambda}
= - \frac{\beta}{\psi}(1+o(1))\, .
\end{align}

Next, the term $-\alpha\log x_0$ may be developed as follows,
\begin{multline*}
-\alpha\log x_0 
\stackrel{\mathrm{(a)}}{=} -\alpha \log\left(
\frac{Q_1^{-1}(1)}{e^{T/\psi}}\beta^{-T/\psi} (1+o(1))\right) \\
\stackrel{\mathrm{(b)}}{=} - \alpha \log \left( \frac{Q_1^{-1}(1)}{e^{T/\psi}}
\left( Q_2(Q_1^{-1}(1)) \right)^{T/\psi} \alpha^{-T/\psi} (1+o(1)) \right) \\
\stackrel{\mathrm{(c)}}{=} - \alpha \! \left[ \log Q_1^{-1}(1) \! + \! \frac{T}{\psi} \! \log
Q_2(Q_1^{-1}(1)) \! - \! \frac{T}{\psi} \! \log\alpha \! - \! \frac{T}{\psi} \! + \! o(1) \right]
\end{multline*}
where we have used \eqref{eq:x0_beta} in $\mathrm{(a)}$, \eqref{eq:beta_alpha} in $\mathrm{(b)}$,
and the Taylor series of $\log(1+x)$ around $x=0$ in $\mathrm{(c)}$. Hence, we
conclude that
\begin{multline}
\label{eq:G_second_term}
-\alpha\log x_0 = \alpha \left( \log\frac{1}{Q_1^{-1}(1)} +
\frac{T}{\psi}\log\frac{1}{Q_2(Q_1^{-1}(1))} \right) \\
+ \frac{T}{\psi} \alpha\log\alpha + \frac{T}{\psi} Q_2(Q_1^{-1}(1)) \beta +o(\alpha)
\end{multline}
where we have again used \eqref{eq:beta_alpha}.

Finally, we analyze the term $\sum_{t\in I_v} \delta_t \log B^{(t)}(x_0,y_0)$. Using
\eqref{eq:y0_beta} and \eqref{eq:x0_beta}, we obtain
\begin{multline*}
B^{(t)}(x_0,y_0) = 1 + \sum_{(i,j) \in S_t^{-}} B^{(t)}_{i,j} x_0^i y_0^j \\
\:\:\:\: = 1 + \sum_{(i,j) \in S_t^{-}} B_{i,j}^{(t)} \left( \frac{Q_1^{-1}(1)}{e^{T/\psi}}
\right)^i C^{\frac{j}{r}} \left(\smallint\!\lambda\right)^{\frac{j}{\psi}} \beta^{\frac{j-iT}{\psi}}
(1+o(1))\, .
\end{multline*}
From $T_{i,j}=\frac{j-\psi}{i} \geq T$ for all $t \in I_v$, $(i,j) \in S_t^{-}$, it follows that
$\frac{j-iT}{\psi}\geq 1$ for all $t \in I_v$, $(i,j) \in S_t^{-}$, with equality if and
only if $t\in Y_v$ and $(i,j) \in P_t$.
Then we have
\begin{equation*}
B^{(t)}(x_0,y_0) \! = \! 1 + \!\! \displaystyle\sum_{(i,j)\in P_t} \!\! B_{i,j}^{(t)} \left( \frac{Q_1^{-1}(1)}{e^{T/\psi}} \right)^i \! C^{\frac{j}{r}}
\left(\smallint\!\lambda\right)^{\frac{j}{\psi}} \beta (1+o(1))
\end{equation*}
when $t \in Y_v$, and $B^{(t)}(x_0,y_0) = 1+o(\beta)$ otherwise.
This implies that 
\begin{equation}
\label{eq:pre_logB}
\log B^{(t)}(x_0,y_0) \! = \!\!\! \displaystyle\sum_{(i,j)\in P_t} \!\! B_{i,j}^{(t)} \left( \frac{Q_1^{-1}(1)}{e^{T/\psi}}
\right)^i \!\! C^{\frac{j}{r}} \left(\smallint\!\lambda\right)^{\frac{j}{\psi}} \beta (1+o(1))
\end{equation}
when $t \in Y_v$, and $\log B^{(t)}(x_0,y_0) = o(\beta)$ otherwise.
This yields
\begin{multline*}
\sum_{t\in I_v} \delta_t \log B^{(t)}(x_0,y_0) \\
\stackrel{\mathrm{(a)}}{=} \beta (1+o(1))
\sum_{t\in Y_v} \delta_t \sum_{(i,j)\in P_t} B_{i,j}^{(t)} \left(
\frac{Q_1^{-1}(1)}{e^{T/\psi}} \right)^i C^{\frac{j}{r}} \left(\smallint\!\lambda\right)^{\frac{j}{\psi}} \\
\stackrel{\mathrm{(b)}}{=} \frac{\beta}{\psi} (1+o(1)) \sum_{t\in Y_v}
\delta_t \!\! \sum_{(i,j)\in P_t} (j-iT) B_{i,j}^{(t)} \left(
\frac{Q_1^{-1}(1)}{e^{T/\psi}} \right)^i \!\! C^{\frac{j}{r}} \left(\smallint\!\lambda\right)^{\frac{j}{\psi}} \\
= \frac{\beta(1+o(1))}{\psi} \sum_{t\in Y_v}
\delta_t \sum_{(i,j)\in P_t} j B_{i,j}^{(t)} \left(
\frac{Q_1^{-1}(1)}{e^{T/\psi}} \right)^i C^{\frac{j}{r}} \left(\smallint\!\lambda\right)^{\frac{j}{\psi}} \\ 
\,\, - \beta \frac{T}{\psi} (1+o(1)) \sum_{t\in Y_v}
\delta_t \sum_{(i,j)\in P_t} i B_{i,j}^{(t)} \left(
\frac{Q_1^{-1}(1)}{e^{T/\psi}} \right)^i C^{\frac{j}{r}} \left(\smallint\!\lambda\right)^{\frac{j}{\psi}} \\
\stackrel{\mathrm{(c)}}{=} \frac{\beta(1+o(1))}{\psi} \sum_{t\in Y_v}
\frac{\lambda_t}{q_t} \sum_{(i,j)\in P_t} j B_{i,j}^{(t)} C^{\frac{j}{r}} \left(
\frac{\smallint\!\lambda}{e} \right)^\frac{iT}{\psi} \!\! \left( Q_1^{-1}(1) \right)^i \\ 
\,\, - \beta \frac{T}{\psi} (1+o(1)) \sum_{t\in Y_v}
\frac{\lambda_t}{q_t} \sum_{(i,j)\in P_t} i B_{i,j}^{(t)} C^{\frac{j}{r}}
\left( \frac{\smallint\!\lambda}{e} \right)^{\frac{iT}{\psi}} \!\! \left( Q_1^{-1}(1) \right)^i
\end{multline*}
where we have used \eqref{eq:pre_logB} in $\mathrm{(a)}$, and $\psi=j-iT$ for $t\in Y_v$ and
$(i,j)\in P_t$ in $\mathrm{(b)}$ and in $\mathrm{(c)}$. Hence, recalling
\eqref{eq:P1x_definition} and \eqref{eq:P2x_definition}, we obtain
\begin{align}
\label{eq:G_first_term}
\sum_{t\in I_v} \delta_t \log B^{(t)}(x_0,y_0) = \frac{\beta(1+o(1))}{\psi} \left( 1 - T\, Q_2(Q_1^{-1}(1)) \right) \, . 
\end{align}
Finally, substituting \eqref{eq:G_fourth_plus_third}, \eqref{eq:G_second_term}, and
\eqref{eq:G_first_term} into \eqref{eq:growth_rate_polynomial_general}, we obtain \eqref{eq:growth_rate_asymptotic_dgldpc1}, as desired.


\section{Closed Form Expressions for the Spectral Shape}\label{appendix:closed_form}
It is worthwhile to note that in some cases, \eqref{eq:G(alpha)_tanner_codes} can be expressed in closed form because $\sff^{-1}(\alpha)$ can be expressed analytically. This is the case, for instance, for the $(3,6)$ regular LDPC ensemble, for which $\sff(z) = \alpha$ becomes $ax^3+bx^2+cx+d=0$, where $x=z^2$ and $(a,b,c,d) = (\alpha-1, 15\alpha-10,15\alpha-5,\alpha)$. This cubic equation in $x$ may be solved by Cardano's method (see, e.g., \cite[p. 17]{wikipedia:Cardanos_method}; the discriminant $\Delta = \rho^3 + \mu^2$ is negative for every $\alpha \in (0,1)$, where 
\[
\rho = \frac{3 a c - b^2}{9 a^2} \; ; \; \mu = \frac{9 a b c - 27a^2 d - 2 b^3}{54 a^3} \; .
\]
The required solution is then uniquely and analytically identified as that given by \eqref{eq:G(alpha)_tanner_codes} where $q=3$, $s=6$ and $\sff^{-1}(\alpha) = z = \sqrt{x}$ where $x = 2 \sqrt{-\rho} \cos \left( \theta/3 \right) - \frac{b}{3a} > 0$ and $\theta = \tan^{-1} \left( \sqrt{-\Delta}/\mu \right)$.

Similarly, the weight spectral shape of a $(4,8)$ regular LDPC ensemble may be expressed in closed form through the solution of a quartic equation.


\begin{thebibliography}{10}
\providecommand{\url}[1]{#1}
\csname url@rmstyle\endcsname
\providecommand{\newblock}{\relax}
\providecommand{\bibinfo}[2]{#2}
\providecommand\BIBentrySTDinterwordspacing{\spaceskip=0pt\relax}
\providecommand\BIBentryALTinterwordstretchfactor{4}
\providecommand\BIBentryALTinterwordspacing{\spaceskip=\fontdimen2\font plus
\BIBentryALTinterwordstretchfactor\fontdimen3\font minus
  \fontdimen4\font\relax}
\providecommand\BIBforeignlanguage[2]{{%
\expandafter\ifx\csname l@#1\endcsname\relax
\typeout{** WARNING: IEEEtran.bst: No hyphenation pattern has been}%
\typeout{** loaded for the language `#1'. Using the pattern for}%
\typeout{** the default language instead.}%
\else
\language=\csname l@#1\endcsname
\fi
#2}}

\bibitem{gallager63:low-density}
R.~Gallager, \emph{Low-Density Parity-Check Codes}.\hskip 1em plus 0.5em minus
  0.4em\relax Cambridge, Massachusetts: M.I.T. Press, 1963.

\bibitem{tanner81:recursive}
R.~M. Tanner, ``A recursive approach to low complexity codes,'' \emph{{IEEE}
  Trans. Inf. Theory}, vol.~27, no.~5, pp. 533--547, Sept. 1981.

\bibitem{liva08:quasi_cyclic}
G.~Liva, W.~E. Ryan, and M.~Chiani, ``Quasi-cyclic generalized {LDPC} codes
  with low error floors,'' \emph{{IEEE} Trans. Commun.}, vol.~56, no.~1, pp.
  49--57, Jan. 2008.

\bibitem{wang06:D-GLDPC}
Y.~Wang and M.~Fossorier, ``Doubly generalized low-density parity-check
  codes,'' in \emph{Proc. of 2006 IEEE Int. Symp. Inf. Theory},
  Seattle, WA, USA, July 2006, pp. 669--673.

\bibitem{miladinovic08:generalized}
N.~Miladinovic and M.~Fossorier, ``Generalized {LDPC} codes and generalized
  stopping sets,'' \emph{{IEEE} Trans. Commun.}, vol.~56, no.~2, pp. 201--212,
  Feb. 2008.

\bibitem{paolini09:stability}
E.~Paolini, M.~Fossorier, and M.~Chiani, ``Doubly-generalized {LDPC} codes:
  Stability bound over the {BEC},'' \emph{{IEEE} Trans. Inf. Theory},
  vol.~55, no.~3, pp. 1027--1046, Mar. 2009.

\bibitem{lentmaier99:generalized}
M.~Lentmaier and K.~Zigangirov, ``On generalized low-density parity-check codes
  based on {H}amming component codes,'' \emph{{IEEE} Commun. Lett.}, vol.~3,
  no.~8, pp. 248--250, Aug. 1999.

\bibitem{boutros99:generalized}
J.~Boutros, O.~Pothier, and G.~Zemor, ``Generalized low density ({T}anner)
  codes,'' in \emph{Proc. of 1999 {IEEE} Int. Conf. Commun.}, vol.~1,
  Vancouver, Canada, June 1999, pp. 441--445.

\bibitem{litsyn02:on_ensembles}
S.~Litsyn and V.~Shevelev, ``On ensembles of low-density parity-check codes:
  asymptotic distance distributions,'' \emph{{IEEE} Trans. Inf. Theory},
  vol.~48, no.~4, pp. 887--908, Apr. 2002.

\bibitem{burshtein04:asymptotic}
D.~Burshtein and G.~Miller, ``Asymptotic enumeration methods for analyzing
  {LDPC} codes,'' \emph{{IEEE} Trans. Inf. Theory}, vol.~50, no.~6, pp.
  1115--1131, June 2004.

\bibitem{orlitsky05:stopping}
A.~Orlitsky, K.~Viswanathan, and J.~Zhang, ``Stopping set distribution of
  {LDPC} code ensembles,'' \emph{{IEEE} Trans. Inf. Theory}, vol.~51, no.~3,
  pp. 929--953, Mar. 2005.

\bibitem{di06:weight}
C.~Di, T.~Richardson, and R.~Urbanke, ``Weight distribution of low-density
  parity-check codes,'' \emph{{IEEE} Trans. Inf. Theory}, vol.~52, no.~11,
  pp. 4839--4855, Nov. 2006.

\bibitem{bennatan04:weight}
A. Bennatan and D. Burshtein, ``On the application of LDPC codes to arbitrary discrete-memoryless channels,'' \emph{{IEEE} Trans. Inf. Theory}, vol.~50, no.~3, pp.~417--437, Mar. 2004.

\bibitem{yang11:weight}
S.~Yang, T.~Honold, Y.~Chen, Z.~Zhang and P.~Qiu, ``Weight distributions of regular low-density
  parity-check codes over finite fields,'' \emph{{IEEE} Trans. Inf. Theory}, vol.~57, no.~11,
  pp. 7507--7521, Nov. 2011.

\bibitem{Kasai2008:nonbinary}
    {K.~Kasai, C.~Poulliat, D.~Declercq, T.~Shibuya, and K.~Sakaniwa,}
    {``Weight distribution of non-binary LDPC codes,''}
    in {\em Proc. 2008 IEEE Int. Symp. Inf. Theory and its Applications,} Auckland, New Zealand, Dec. 2008, pp. 1--6.

\bibitem{andriyanova09:weight}
I. Andriyanova, V. Rathi and J.-P. Tillich, ``Binary weight distribution of non-binary LDPC codes'', \emph{Proc. 2009 IEEE Int. Symp. Inf. Theory}, Seoul, Korea, June/July 2009, pp.~65--69.

\bibitem{divsalar06:weight_enumerators}
D.~Divsalar, ``Ensemble weight enumerators for protograph {LDPC} codes,'' in
  \emph{Proc. 2006 {IEEE} Int. Symp. Inf. Theory}, Seattle, WA, USA, July
  2006, pp. 1554--1558.

\bibitem{wang2009:multiedge}
    {C.~-Li Wang, S.~Lin and M.~P.~C.~Fossorier,}
    {``On asymptotic ensemble weight enumerators of multi-edge type codes,''}
    in \emph{Proc. 2009 IEEE Global Telecommun. Conf.}, Honolulu, HI, USA, Nov./Dec. 2009.

\bibitem{Kasai2009:multiedge}
    {K.~Kasai, T.~Awano, D.~Declercq, C.~Poulliat and K.~Sakaniwa,}
    {``Weight distributions of multi-edge type LDPC codes,''}
    in {\em Proc. 2009 IEEE Int. Symp. Inf. Theory,} Seoul, Korea, June/July 2009, pp. 60--64.

\bibitem{abu-surra11:IEEE-IT}
S.~Abu-Surra, D.~Divsalar, and W.~E. Ryan, ``Enumerators for protograph-based ensembles of {LDPC}
and generalized {LDPC} codes,''
  \emph{{IEEE} Trans. Inf. Theory}, vol. 57, no. 2, pp. 858--886, Feb. 2011.

\bibitem{wang08:ensemble_DGLDPC}
Y.~Wang, C.-L. Wang, and M.~Fossorier, ``Ensemble weight enumerators for
  protograph-based doubly generalized {LDPC} codes,'' in \emph{Proc. 2008
  {IEEE} Int. Symp. Inf. Theory}, Toronto, Canada, July 2008, pp.
  1168--1172.

\bibitem{paolini09:class}
E.~Paolini, M.~F. Flanagan, M.~Chiani, and M.~Fossorier, ``On a class of
  doubly-generalized {LDPC} codes with single parity-check variable nodes,'' in \emph{Proc. 2009
  {IEEE} Int. Symp. Inf. Theory}, Seoul, Korea, July 2009, pp. 1983--1987.

\bibitem{flanagan09:IEEE-IT}
M.~F.~Flanagan, E.~Paolini, M.~Chiani, and M.~Fossorier, ``On the growth rate of
  the weight distribution of irregular doubly-generalized {LDPC} codes,''
  \emph{{IEEE} Trans. Inf. Theory}, vol. 57, no. 6, pp. 3721--3737, June 2011.

\bibitem{di02:finite}
    {C.~Di, D.~Proietti, I.~E. Telatar, T.~J. Richardson and R.~Urbanke,}
    {``Finite-length analysis of low-density parity-check codes on the binary erasure channel,''}
    {\em IEEE Trans. Inf. Theory,} vol.~48, no.~6, pp. 1570--1579, June~2002.

\bibitem{pless:intro_to_theory_of_ecc}
V. Pless, \emph{Introduction to the theory of error-correcting codes}.\hskip 1em plus 0.5em minus 0.4em\relax New York: John Wiley and Sons, 1998.

\bibitem{measson08:maxwell_construction}
C. M\'{e}asson, A. Montanari and R. L. Urbanke, ``Maxwell construction: The hidden bridge between iterative and maximum \emph{a posteriori} decoding'', \emph{IEEE Transactions on Information Theory}, vol.~54, no.~12, pp.~5277--5307, Dec. 2008.

\bibitem{storn05:differential-book}
K.~Price, R.~Storn, and J.~Lampinen, \emph{Differential Evolution: A Practical
  Approach to Global Optimization}, 1st~ed.\hskip 1em plus 0.5em minus
  0.4em\relax Berlin, Germany: Springer-Verlag, 2005.

\bibitem{citeulike:3401906}
T.~Back, D.~Fogel, and Z.~Michalewicz, Eds., \emph{Handbook of Evolutionary
Computation}.\hskip 1em plus 0.5em minus 0.4em\relax Bristol, UK: IOP
Publishing Ltd., 1997.

\newpage
\bibitem{shokrollahi00:design}
M.~Shokrollahi and R.~Storn, ``Design of efficient erasure codes with
differential evolution,'' in \emph{Proc. of 2000 {IEEE} Int. Symp.
Inf. Theory}, Sorrento, Italy, June 2000, p. 5.

\bibitem{Tillich04:weight}
J.~-P.~Tillich, ``The average weight distribution of {T}anner code ensembles and a way to modify them to improve
their weight distribution,'' in \emph{Proc. 2004 IEEE Int. Symp. Inf. Theory}, Chicago, IL, USA, June/July 2004, pp. 7.

\bibitem{paolini08:weight}
E.~Paolini, M.~Chiani and M.~Fossorier, ``On the growth rate of GLDPC codes weight distribution,'' in \emph{Proc. 2008 IEEE
Int. Symp. Spread Spectrum Techniques and Applications}, Bologna, Italy, Aug.~2008, pp. 790--794.

\bibitem{wikipedia:Cardanos_method}
M. Abramowitz and I. A. Stegun, \emph{Handbook of mathematical functions with formulas, graphs, and mathematical tables}.\hskip 1em plus 0.5em minus 0.4em\relax New York: Dover, 1964.

\end{thebibliography}
\end{document}